\documentclass[english,final]{IEEEtran}
\usepackage[T1]{fontenc}
\usepackage[latin9]{inputenc}
\usepackage{xcolor}
\usepackage{mathrsfs}
\usepackage{bm}
\usepackage{amsmath}
\usepackage{amsthm}
\usepackage{amssymb}
\usepackage{stmaryrd}
\usepackage{graphicx}
\PassOptionsToPackage{normalem}{ulem}
\usepackage{ulem}

\makeatletter
\theoremstyle{definition}
\newtheorem{defn}{\protect\definitionname}
\theoremstyle{plain}
\newtheorem{thm}{\protect\theoremname}
\theoremstyle{plain}
\newtheorem{prop}{\protect\propositionname}
\theoremstyle{plain}
\newtheorem{lem}{\protect\lemmaname}
\theoremstyle{remark}
\newtheorem{rem}{\protect\remarkname}

\usepackage{amsmath}
\usepackage{amssymb}

\usepackage{graphicx,psfrag,cite,subfigure}
\author{

\IEEEauthorblockN{Junting~Chen and Urbashi~Mitra}

\IEEEauthorblockA{Ming Hsieh Department of Electrical Engineering, 
University of Southern California \\
Los Angeles, CA 90089 USA, email:\{juntingc, ubli\}@usc.edu}

\thanks{This research has been funded in part by one or more of the following grants: ONR N00014-15-1-2550, NSF CNS-1213128, NSF CCF-1718560, NSF CCF-1410009, NSF CPS-1446901, and AFOSR FA9550-12-1-0215.}
}


\usepackage[acronym]{glossaries}
\newcommand{\newac}{\newacronym}
\newcommand{\ac}{\gls}
\newcommand{\Ac}{\Gls}

\makeglossaries

\newac{speb}{SPEB}{square position error bound}
\newac[plural=EFIMs,firstplural=Fisher information matrices (EFIMs)]{efim}{EFIM}{Fisher information matrix}
\newac{ne}{NE}{Nash equilibrium}
\newac{mse}{MSE}{mean squared error}
\newac{toa}{TOA}{time-of-arrival}
\newac{snr}{SNR}{signal-to-noise ratio}
\newac{lan}{LAN}{local area network}
\newac{psd}{PSD}{positive semidefinite}
\newac{pd}{PD}{positive definite}
\newac{wrt}{w.r.t.}{with respect to}
\newac{lhs}{L.H.S.}{left hand side}
\newac{wp1}{w.p.1}{with probability 1}
\newac{kkt}{KKT}{Karush-Kuhn-Tucker}
\newac{wlog}{w.l.o.g.}{without loss of generality}
\newac{mle}{MLE}{maximum likelihood estimation}
\newac{gps}{GPS}{global positioning system}
\newac{rssi}{RSSI}{received signal strength}
\newac{mimo}{MIMO}{multiple-input multiple-output}
\newac{csi}{CSI}{channel state information}
\newac{fdd}{FDD}{frequency division duplexing}
\newac{ms}{MS}{mobile station}
\newac{bs}{BS}{base station}
\newac{d2d}{D2D}{device-to-device}
\newac{slnr}{SLNR}{signal-to-interference-leakage-and-noise-ratio}
\newac{ula}{ULA}{uniform linear antenna array}
\newac{pas}{PAS}{power angular spectrum}
\newac{mmse}{MMSE}{minimum mean square error}
\newac{zf}{ZF}{zero-forcing}
\newac{rzf}{RZF}{regularized zero-forcing}
\newac{as}{AS}{angular spread}
\newac{aod}{AOD}{angle of departure}
\newac{iid}{i.i.d.}{independent and identically distributed} 
\newac{sinr}{SINR}{signal-to-interference-and-noise ratio}
\newac{tdd}{TDD}{time-division duplex}
\newac{rvq}{RVQ}{random vector quantization}
\newac{rhs}{R.H.S.}{right hand side}
\newac{mrc}{MRC}{maximum ratio combining}
\newac{cdf}{CDF}{cumulative distribution function}
\newac{a.s.}{a.s.}{almost surely}
\newac{los}{LOS}{line-of-sight}
\newac{jsdm}{JSDM}{joint spatial division and multiplexing}
\newac{map}{MAP}{maximum a posteriori}
\newac{klt}{KLT}{Karhunen-Lo\`eve Transform}
\newac{lbe}{LBE}{link bargaining equilibrium}
\newac{se}{SE}{Stackelberg equilibrium}
\newac{uav}{UAV}{unmanned aerial vehicle}
\newac{nlos}{NLOS}{non-line-of-sight}
\newac{pdf}{PDF}{probability density function}
\newac{em}{EM}{expectation-maximization}
\newac{knn}{KNN}{$k$-nearest neighbor}
\newac{svd}{SVD}{singular value decomposition}
\newac{nmf}{NMF}{non-negative matrix factorization}
\newac{umf}{UMF}{unimodality-constrained matrix factorization}
\newac{rmse}{RMSE}{rooted mean squared error}
\newac{olos}{OLOS}{obstructed line-of-sight}
\newac{mmw}{mmW}{millimeter wave}
\newac{ber}{BER}{bit error rate}
\newac{rss}{RSS}{received signal strength}
\newac{lp}{LP}{linear program}
\newac{ufw}{U-FW}{unimodal Frank-Wolfe}
\newac{utf}{UTF}{unimodality-constrained tensor factorization}
\newac{fw}{FW}{Frank-Wolfe}

\setkeys{Gin}{width=1.0\columnwidth}

\makeatother

\usepackage{babel}
\providecommand{\definitionname}{Definition}
\providecommand{\lemmaname}{Lemma}
\providecommand{\propositionname}{Proposition}
\providecommand{\remarkname}{Remark}
\providecommand{\theoremname}{Theorem}

\begin{document}
\title{Unimodality-Constrained Matrix Factorization for Non-Parametric Source
Localization}

\maketitle
%% The following is the formatting requirement of TWC submission
%
% 12pt, draftclsnofoot, peerreview, a4paper, oneside, onecolumn
%
% ================
%% The following commnds automatically adjust the size of the figures and the font size on the figures according to single/double column confirguration. However, it does not affect the figures with their size exiplicitly specified. It also requrie \figfontsize command to be put in the psfrag block.
%% Choice of font size: tiny, scriptsize, footnotesize, small, normalsize, large, Large, LARGE, huge Huge

%% Uncomment the following if a single column format is used.

%\newcommand*{\SINGLECOLUMN}{}

\ifdefined\SINGLECOLUMN
	% - single column setting
	\setkeys{Gin}{width=0.5\columnwidth}
	\newcommand{\figfontsize}{\footnotesize} 
\else
	% - double column setting
	\setkeys{Gin}{width=1.0\columnwidth}
	\newcommand{\figfontsize}{\normalsize} 
\fi
% ================
%% The follwoing define annotation formate
% TBD
\begin{abstract}
Herein, the problem of simultaneous localization of multiple sources
given a number of energy samples at different locations is examined.
The strategies do not require knowledge of the signal propagation
models, nor do they exploit the spatial signatures of the source.
A non-parametric source localization framework based on a matrix observation
model is developed. It is shown that the source location can be estimated
by localizing the peaks of a pair of location signature vectors extracted
from the incomplete energy observation matrix. A robust peak localization
algorithm is developed and shown to decrease the source localization
\ac{mse} faster than $\mathcal{O}(1/M^{1.5})$ with $M$ samples,
when there is no measurement noise. To extract the source signature
vectors from a matrix with mixed energy from multiple sources, a \ac{umf}
problem is formulated, and two rotation techniques are developed to
solve the \ac{umf} efficiently. Our numerical experiments demonstrate
that the proposed scheme achieves similar performance as the kernel
regression baseline using only $1/5$ energy measurement samples in
detecting a single source, and the performance gain is more significant
in the cases of detecting multiple sources.
\end{abstract}

\glsresetall

\section{Introduction}

\label{sec:intro}

%% Points of arguement:
% - computational complexity (the problem seems to be easy, but actually not)
% - long delay multilage channel such that we cannot resolve the paths and the sources

Source localization is important in many domains, such as salvage,
exploration, tactical surveillance, and hazard finding. However, in
many application scenarios, it is difficult to obtain the correct
propagation parameters of the source signal for localization. For
example, in underwater localization with acoustic signals, the signal
propagation depends on the water temperature, pressure, and salinity,
which are location-dependent. In gas source localization, the gas
diffusion characteristics depends on the chemical type and the atmospheric
conditions. Therefore, model-based parametric localization methods
\cite{BecStoLi:J08,SheHu:J05,MeeMitNar:J08,LiuZakChe:J12} may not
be reliable in application scenarios with a temporal and spatial varying
nature. 

Model-free positioning schemes, such as connectivity based localizations
and weighted centroid localizations (WCL), have attracted a lot of
interest due to their simplicity in implementation and the robustness
to variations of propagation properties \cite{chen2010mobile,he2003range,wang2011weighted,ChoMit:C15,CheMit:C17a}.
However, connectivity based techniques \cite{chen2010mobile,he2003range}
can only provide coarse localization results and the performance of
WCL \cite{wang2011weighted} highly depends on the choice of parameters
and the propagation environments. Recently, machine learning techniques,
such as kernel regression and support vector machines \cite{JinSohWon:J10,KimParYooKimPar:J13}
have also been explored for localization. However, these methods usually
require a separate training phase which may not be available in practice.
Furthermore, blind deconvolution methods for source separation and
localization \cite{abadi2012blind,ahmed2014blind,li2016off,tang2013compressed,yang2016super,cichocki2006new}
usually require a sequence of measurements that convey temporal or
spatial characteristics, which is not the case in our problem.

\begin{figure}
\begin{centering}
\includegraphics[width=1\columnwidth]{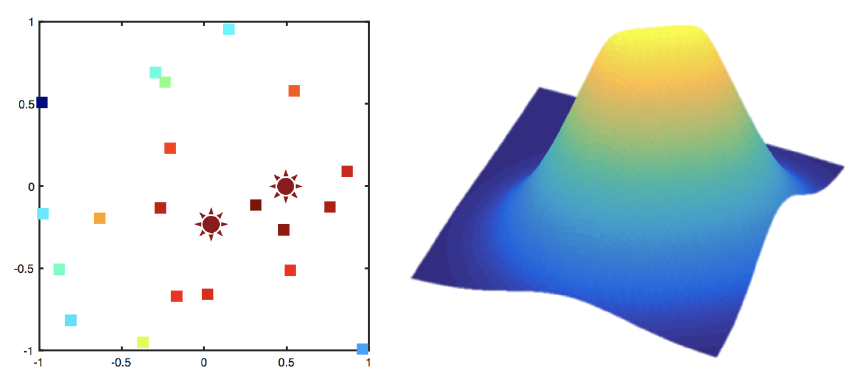}
\par\end{centering}
\caption{\label{fig:model} Left: Energy measurements at different locations
(colored-squares) to localize two sources (red stars). Right: The
underlying energy field appears as one peak when the two sources are
close to each other. }
\end{figure}
This paper studies non-parametric methods for localizing several sources
based on a few energy measurements at different sensing locations
as illustrated in Fig. \ref{fig:model}. Our previous works studied
the single source case in \cite{ChoMit:C15,ChoKumNarMit:J16}, where
a trust region was developed for targeting the source and a multi-step
exploration-exploitation strategy was developed for active search
using an underwater robot. The results were extended to the cases
of two or more sources by exploiting novel coordinate system rotation
techniques \cite{CheMit:C17a,CheMit:C17b}. However, these works were
based on a decomposable assumption on the energy propagation field
in 2D.

In this paper, we find that the decomposable assumption is \emph{not}
necessary. Specifically, we show that the source location in 2D can
be found by localizing the peaks of the left and right dominant singular
vectors of the energy matrix sampled on a discretized area. Such a
result holds \emph{universally} as long as the received signal energy
is a decreasing function of the distance from the source. With such
a theoretical guarantee, we first develop strategies to localize a
single source. The method extracts a pair of signature vectors from
a sparsely sampled observation matrix, and then, estimates the source
location by robust peak localization from the signature vectors. The
corresponding localization \ac{mse} is also analyzed. We then move
to the case of localizing two sources with equal power. The \ac{svd}
framework in our preliminary work \cite{ChoKumNarMit:J16} does not
work herein, because the singular vectors are not the desired signature
vectors of each source. We address this issue by optimally rotating
the coordinate system such that the sources are aligned with the rows
or columns of the observation matrix, and as a result, the \ac{svd}-based
method can be applied. Finally, we consider the general multi-source
case by formulating a \ac{umf} problem and solve it with projected
gradient algorithms. While a formal theoretical justification that
connects the multi-source localization with the \ac{umf} is yet to
be explored, our experiments demonstrate the robustness of the proposed
\ac{umf}-based methods.

To summarize, the following contributions are made herein: 
\begin{itemize}
\item We propose a matrix observation model for the energy field of a single
source, and prove that the left and right dominant singular vectors
are unimodal with their peaks representing the source location. 
\item We develop non-parametric localization algorithms based on sparse
matrix processing. In the single source case, we derive a localization
\ac{mse} bound and show that the \ac{mse} decreases faster than
$\mathcal{O}(1/M^{1.5})$ when there is no sampling noise.
\item For a general multi-source case, we formulate a \ac{umf} problem
and develop a projected gradient algorithm to extract the signature
vectors for localizing each of the sources. Our numerical experiments
demonstrate the robustness of the these methods in uncertain environments. 
\end{itemize}

The rest of the paper is organized as follows. Section \ref{sec:system-model}
establishes the matrix observation model. Section \ref{sec:localization-via-unimodal-symmetric}
develops a non-parametric localization estimator by establishing the
unimodal and symmetry property of the source signature vectors. Section
\ref{sec:single-source} extracts the signature vectors for the single
source scenario and analyzes the \ac{mse} localization performance.
Section \ref{sec:double-source} and \ref{sec:arbitrary-number-of-sources}
develop coordinate system rotation techniques to help extract the
signature vectors in the two source and arbitrary number of sources
cases. Numerical results are presented in Section \ref{sec:numerical}
and conclusion is given in Section \ref{sec:conclusion}. 

\emph{Notation:} Vectors are written as bold italic letters $\bm{\mathit{x}}$
and matrices as bold capital italic letters $\bm{\mathit{X}}$. Random
variables, random vectors, and random matrices are written as $\mathsf{x}$,
bold letters $\bm{\mathsf{x}}$, and bold capital letters $\bm{\mathsf{X}}$,
respectively. For a matrix $\bm{X}$, $X_{ij}$ denotes the entry
in the $i$th row and $j$th column of $\bm{X}$. For a vector $\bm{x}$,
$x_{i}$ denotes the $i$th entry of $\bm{x}$. The notation $o(x)$
means $\lim_{x\to0}o(x)/x\to0$, and $\mathcal{O}(x)$ means $\lim\sup_{x\to0}\mathcal{O}(x)/x<\infty.$

\section{System Model}

\label{sec:system-model}

% Do we need a figure here? Think about it.

Assume that there are $K$ incoherent sources in an area with radius
$L/4$. The location of source $k$ is denoted by $\bm{\mathit{s}}_{k}\in\mathbb{R}^{2}$.
The sources continuously emit signals that form an aggregated energy
field, which can be measured at $M$ different sensing locations in
the target area $\mathcal{A}$ with radius $L/2$. Let $d(\bm{\mathit{z}},\bm{\mathit{s}})=\|\bm{\mathit{z}}-\bm{s}\|_{2}$
be the distance between the sensing location at $\bm{\mathit{z}}\in\mathbb{R}^{2}$
and a source location $\bm{\mathit{s}}$. The energy measurement $\mathsf{h}^{(m)}$
at the $m$th sensing location $\bm{\mathit{z}}^{(m)}$ is given by
\begin{equation}
\mathsf{h}^{(m)}=\sum_{k=1}^{K}\alpha_{k}h(d(\bm{\mathit{z}}^{(m)},\bm{\mathit{s}}_{k}))+\mathsf{n}^{(m)}\label{eq:measurement-model}
\end{equation}
where $\alpha_{k}$ is the transmit power of source $k$ and $\mathsf{n}^{(m)}$
is the additive noise of the $m$th measurement. The measurement noise
$\mathsf{n}^{(m)}$ is modeled as a zero mean random variable with
variance $\sigma_{\text{n}}^{2}$ and bounded support $|\mathsf{n}^{(m)}|<\bar{\sigma}_{\text{n}}$.
The function $h(d)$ is a non-negative strictly decreasing function
of the distance $d$ from the source. In addition, we assume that
$h(d)$ is Lipschitz continuous, square-integrable, and concentrated
in the bounded area $\mathcal{A}$, \emph{i.e.}, $\int_{\mathbb{R}^{2}}h(d(\bm{z},\bm{s}))^{2}d\bm{z}=\int_{\mathcal{A}}h(d(\bm{z},\bm{s}))^{2}d\bm{z}=1$.
Note that, in reality, the effective energy response $h(d)$ measured
by practical devices always has a bounded support. However, neither
the source power $\alpha_{k}$, the function $h(d)$, nor the noise
distribution are known.

We propose to use a matrix observation model for non-parametric source
localization. Specifically, we discretize the $L\times L$ target
area, $\mathcal{A}$, into $N\times N$ grid points, equally spaced
with minimum distance $\frac{L}{N}$, $N\geq\sqrt{M}$. Let $\bm{\mathsf{H}}$
be the $N\times N$ observation matrix that contains the $M$ energy
measurements, \emph{i.e.}, the $(i,j)$th entry of $\bm{\mathsf{H}}$
is given by 
\begin{equation}
\mathsf{H}_{ij}=\frac{L}{N}\mathsf{h}^{(m)}\label{eq:matrix-observation-model}
\end{equation}
if the $m$th energy measurement is taken inside the $(i,j)$th grid
cell. Note that $\bm{\mathsf{H}}$ may contain missing values and
the measurement locations $\bm{z}^{(m)}$ may not be at the center
of the grid cell. Our objective herein is to localize the sources
by analyzing the incomplete matrix $\bm{\mathsf{H}}.$

\textcolor{purple}{}

\section{Localization based on Unimodality}

\label{sec:localization-via-unimodal-symmetric}

In this section, we first show that the dominant singular vectors
of the energy matrix sampled in a discretized single source energy
field are unimodal and symmetric. Then, using these properties, a
localization algorithm is developed. Finally, a matrix factorization
problem is formulated to extract the signature vectors in the case
of multiple sources and incomplete matrix observations. 

\subsection{The Unimodal Property}

Let $\bm{\mathit{H}}^{(k)}\in\mathbb{R}^{N\times N}$ be the matrix
that captures the energy contributed by source $k$ sampled at discretized
locations $\{\bm{c}_{i,j}\}$, \emph{i.e.}, the $(i,j)$th element
of $\bm{\mathit{H}}^{(k)}$ is given by 
\begin{equation}
H_{ij}^{(k)}=\frac{L}{N}\alpha_{k}h(d(\bm{\mathit{c}}_{i,j},\bm{\mathit{s}}_{k}))\label{eq:matrix-observation-model-true}
\end{equation}
in which, $\bm{\mathit{c}}_{i,j}\in\mathbb{R}^{2}$ is the center
location of the $(i,j)$th grid cell. The factor $\frac{L}{N}$ is
for normalization purposes: 
\begin{align}
\big\|\bm{H}^{(k)}\big\|_{\text{F}}^{2} & =\alpha_{k}^{2}\sum_{i,j}\Big(\frac{L}{N}\Big)^{2}h(d(\bm{\mathit{c}}_{i,j},\bm{\mathit{s}}_{k}))^{2}\nonumber \\
 & =\alpha_{k}^{2}\int_{\mathcal{A}}h(d(\bm{z},\bm{\mathit{s}}_{k}))^{2}d\bm{z}+o\Big(\frac{L^{2}}{N^{2}}\Big)\label{eq:Hk-normalization}
\end{align}
which equals to $\alpha_{k}^{2}$ up to a marginal discretization
error. In the signal model (\ref{eq:matrix-observation-model-true})
and (\ref{eq:Hk-normalization}), the physical meaning of $h(d(x,y))$
is the power density of the source signal measured at location $(x,y)$
and the entry $H_{ij}^{(k)}$ approximates the energy of the $k$th
source signal in the $(i,j)$th grid cell. 

Consider the \ac{svd} of $\bm{H}^{(k)}=\sum_{i}\sigma_{k,i}\bm{\mathit{u}}_{k,i}\bm{\mathit{v}}_{k,i}^{\text{T}}$,
where $\sigma_{k,i}$ denote the $i$th largest singular value of
$\bm{H}^{(k)}$. Then, we have the following model for the $K$ source
energy field.

\begin{defn}
[Signature vector and signature matrix]\label{def:Signature-vector}
The \emph{signature matrix} for all $K$ sources is defined as
\begin{align}
\bm{\mathit{H}}\triangleq\sum_{k=1}^{K}\bm{\mathit{H}}^{(k)} & =\sum_{k=1}^{K}\sigma_{k,1}\bm{u}_{k,1}\bm{v}_{k,1}^{\text{T}}+\sum_{k=1}^{K}\sum_{i=2}^{N}\sigma_{k,i}\bm{\mathit{u}}_{k,i}\bm{\mathit{v}}_{k,i}^{\text{T}}.\label{eq:matrix-decomposition-model}
\end{align}
In addition, the vectors $\bm{u}_{k}\triangleq\bm{u}_{k,1}$ and $\bm{v}_{k}\triangleq\bm{v}_{k,1}$,
are defined as the \emph{signature vectors} of source $k$. 
\end{defn}

Note that $\bm{u}_{k}$ and $\bm{v}_{k}$ are not singular vectors
of $\bm{H}$, but they are the dominant singular vectors of $\bm{H}^{(k)}$,
which captures the energy contribution from the $k$th source.

Accordingly, the observation matrix $\bm{\mathsf{H}}$ in (\ref{eq:matrix-observation-model})
is a noisy and incomplete sampled version of the signature matrix
$\bm{H}$.

We show that the signature vectors $\bm{u}_{k}$ and $\bm{v}_{k}$
are \emph{unimodal}.
\begin{defn}
[Unimodal] A vector $\bm{v}\in\mathbb{R}^{N}$ is unimodal if the
following is satisfied: 
\begin{align}
0\leq & v_{1}\leq v_{2}\leq\dots\leq v_{s}\label{eq:unimodal-1}\\
 & v_{s}\geq v_{s+1}\geq\dots\geq v_{N}\geq0\label{eq:unimodal-2}
\end{align}
 for some integer $1\leq s\leq N$, where $v_{i}$ is the $i$th entry
of $\bm{v}$. 
\end{defn}

Note that there could be multiple entries $v_{s}=v_{s+1}=\dots=v_{s+I}$
that are the largest. In other words, the vector has a flat peak.
This will not affect the algorithm design, nor the analytical results
in this paper. 
\begin{thm}
[Unimodal Signature Vector]\label{thm:unimodal-signature-vector}
The signature vectors $\bm{u}_{k}$ and $\bm{v}_{k}$ are unimodal.
In addition, suppose that source $k$ is located inside the $(m,n)$th
grid cell. Then the peaks of $\bm{u}_{k}$ and $\bm{v}_{k}$ are located
at the $m$th entry of $\bm{u}_{k}$ and the $n$th entry of $\bm{v}_{k}$,
respectively. 
\end{thm}
\begin{proof}
See Appendix \ref{app:pf-unimodal-sinature-vector}.
\end{proof}

Note that such a property holds for general propagation functions
$h(d)$.

\subsection{The Symmetry Property}

It can be further shown that the signature vectors $\bm{u}_{k}$ and
$\bm{v}_{k}$ are symmetric. 

Let $\bm{u}_{k}^{N}$ and $\bm{v}_{k}^{N}$ denote the signature vectors
under the $N\times N$ grid topology. Consider a Cartesian coordinate
system $\mathcal{C}$ and denote $\bm{c}_{i,j}=(c_{i,j,1},c_{i,j,2})\in\mathbb{R}^{2}$
as the coordinates of the $(i,j)$th grid point, where every row of
grid points $\{\bm{c}_{i,1},\bm{c}_{i,2},\dots,\bm{c}_{i,N}\}$ has
the same set of $x$-coordinates $\bm{c}_{\text{X}}=[c_{\text{X},1},c_{\text{X},2},\dots,c_{\text{X},N}]$,
and every column of grids $\{\bm{c}_{1,j},\bm{c}_{2,j},\dots,\bm{c}_{N,j}\}$
has the same set of $y$-coordinates $\bm{c}_{\text{Y}}=[c_{\text{Y},1},c_{\text{Y},2},\dots,c_{\text{Y},N}]$.
Let $\bar{u}_{k}^{N}(y)$ and $\bar{v}_{k}^{N}(x)$ be linearly interpolated
functions from vectors $\bm{u}_{k}^{N}$ and $\bm{v}_{k}^{N}$, respectively.
Specifically, $\bar{u}_{k}^{N}(c_{\text{Y},i})=\sqrt{N/L}u_{k,i}^{N}$,
the $i$th entry of $\bm{u}_{k}^{N}$, and $\bar{v}_{k}^{N}(c_{\text{X},j})=\sqrt{N/L}v_{k,j}^{N}$,
the $j$th entry of $\bm{v}_{k}^{N}$. The off-grid values of $\bar{u}_{k}^{N}(y)$
and $\bar{v}_{k}^{N}(x)$ are obtained through linear interpolation.
Then, we can show the following property.

\begin{prop}
[Symmetry Property]\label{prop:symmetry-signature-vector} Suppose
that there exist squared-integrable functions $w_{k,1}(y)$ and $w_{k,2}(x)$,
such that $\bar{u}_{k}^{N}(y)$ and $\bar{v}_{k}^{N}(x)$ uniformly
converge to $w_{k,1}(y)$ and $w_{k,2}(x)$, respectively, as $N\to\infty$.
Then, $w_{k,1}(y)$ and $w_{k,2}(x)$ can be expressed as $w_{k,1}(y)=w(y-s_{k,2})$
and $w_{k,2}(x)=w(y-s_{k,1})$, where $w(x)$ is a symmetric, unimodal,
non-negative function with $w(-x)=w(x)$, and $\int_{-\infty}^{\infty}w(x)^{2}dx=1$.
\end{prop}
\begin{proof}
See Appendix \ref{app:pf-homogenous}. 
\end{proof}

Proposition \ref{prop:symmetry-signature-vector} suggests a method
to find the peaks of the signature vectors $\bm{u}_{k}$ and $\bm{v}_{k}$
using the symmetry property.

\subsection{A Location Estimator}

We first establish a general property of a unimodal symmetric function
$w(x)$.
\begin{lem}
[Monotone property] \label{lem:Monotonicity} Suppose that the non-negative
function $w(x)$ is unimodal and symmetric about $x=0$. Then, the
autocorrelation function 
\begin{equation}
\tau(t)=\int_{-\infty}^{\infty}w(x)w(x-t)dx\label{eq:autocorrelation-function}
\end{equation}
is non-negative and symmetric about $t=0$. In addition, $\tau(t)$
is strictly decreasing in $t>0$.
\end{lem}
\begin{proof}
The result can be easily derived using the unimodal and symmetry property
of $w(x)$. The details are omitted here due to page limit.
\end{proof}

From Lemma \ref{lem:Monotonicity}, $\tau(t)$ is maximized as $t=0$.
As a result, the non-negative, unimodal, and symmetric function $w(y-s_{1,1})$
from Proposition \ref{prop:symmetry-signature-vector} has the following
autocorrelation function 
\begin{align}
 & \int_{-\infty}^{\infty}w(y-s_{1,1})w(-y+t-s_{1,1})dy\label{eq:reflected-correaltion-w}\\
 & \qquad=\int_{-\infty}^{\infty}w(y-s_{1,1})w(y-t+s_{1,1})dy\label{eq:uu-first-equality}\\
 & \qquad=\int_{-\infty}^{\infty}w(z)w(z+2s_{1,1}-t)dz\label{eq:uu-2nd-equality}\\
 & \qquad=\tau(t-2s_{1,1})\nonumber 
\end{align}
which is maximized at $t=2s_{1,1}$, where the first equality (\ref{eq:uu-first-equality})
is due to symmetry $w(y)=w(-y)$, and the second equality (\ref{eq:uu-2nd-equality})
is from the change of variable $z=y-s_{1,1}$. 

Given a vector $\bm{\mathsf{v}}\in\mathbb{R}^{N}$ and the corresponding
$N$-point coordinates $\bm{c}_{\text{X}}=[c_{\text{X},1},c_{\text{X},2},\dots,c_{\text{X},N}]$,
let $\bar{v}(x)$ be an interpolation function such that $\bar{v}(x)=\sqrt{N/L}\mathsf{v}_{i}$
for $x=c_{\text{X},i}$, $1\leq i\leq N$, and $\bar{v}(x)$ is equivalent
to a linear interpolation at off-grid locations. Define the \emph{reflected
correlation function} as 
\begin{equation}
R(t;\bm{\mathsf{v}},\bm{c}_{\text{X}})=\int_{-\infty}^{\infty}\bar{v}(x)\bar{v}(-x+t)dx.\label{eq:reflected-correlation}
\end{equation}
Then, an estimate of the point of symmetry for $\bm{\mathsf{v}}$
in a continuous interval can be obtained as 
\begin{equation}
\hat{\mathsf{s}}(\bm{\mathsf{v}})=\frac{1}{2}\underset{t\in\mathbb{R}}{\text{argmax}}\,R(t;\bm{\mathsf{v}},\bm{c}_{\text{X}}).\label{eq:location-estimator}
\end{equation}

Note that, the estimator (\ref{eq:location-estimator}) aggregates
the contributions from all measurements, including those far away
from the source. 

As a result, if one can obtain estimates $\hat{\bm{\mathsf{u}}}_{k}$
and $\hat{\bm{\mathsf{v}}}_{k}$ of the signature vectors $\bm{u}_{k}$
and $\bm{v}_{k}$ from the observation matrix $\bm{\mathsf{H}}$,
the estimate of the $k$th the source location can be computed as
$\hat{\bm{\mathsf{s}}}_{k}=(\hat{\mathsf{s}}(\hat{\bm{\mathsf{\bm{v}}}}_{k}),\hat{\mathsf{s}}(\hat{\bm{\mathsf{u}}}_{k}))$
according to the symmetry property in Proposition \ref{prop:symmetry-signature-vector}.

\subsection{Problem Formulation for Extracting Signature Vectors}

There are two remaining issues: First, one needs to extract $K$ pairs
of signature vectors from the incomplete noisy observation matrix
$\bm{\mathsf{H}}$. Second, one needs to find the best coordinate
system $\mathcal{C}$ for grid points $\bm{c}_{\text{X}}\times\bm{c}_{\text{Y}}$
that define the observation matrix $\bm{\mathsf{H}}=\bm{\mathsf{H}}(\theta)$,
since $\bm{\mathsf{H}}(\theta)$ is variant according to the rotation
$\theta$ of the coordinate system. 

We answer these two questions by proposing a {\em \ac{umf}} problem
specified as follows. 

Denote by $\mathcal{U}_{s}^{N}$ the cone specified by the unimodal
constraints (\ref{eq:unimodal-1})--(\ref{eq:unimodal-2}) for a
fixed $s$. Denote $\mathcal{U}^{N}=\bigcup_{s=1}^{N}\mathcal{U}_{s}^{N}$
as the non-negative unimodal cone, and $\mathcal{U}^{N\times K}$
as the set of $N\times K$ real matrices where all the columns are
in $\mathcal{U}^{N}$. Let $\bm{U},\bm{V}\in\mathbb{R}^{N\times K}$
be the matrices that each contains $K$ pairs of signature vectors
$\{\bm{u}_{k},\bm{v}_{k}\}$ to be determined. Let $\bm{W}\in\mathbb{R}^{N\times N}$
be an indicator matrix that describes the sampling strategy, where
$W_{ij}=1$, if $(i,j)\in\Omega$, and $W_{ij}=0$, otherwise, where
$\Omega$ denotes the set of entries that are assigned values based
on (\ref{eq:matrix-observation-model}), $|\Omega|=M$. 

The source signature vectors can be extracted as the solution to the
following problem: 
\begin{align}
\mathscr{P}1:\quad\underset{\bm{U},\bm{V}}{\text{minimize}} & \quad\big\|\bm{W}\varodot\big(\bm{\mathsf{H}}-\bm{U}\bm{V}^{\text{T}}\big)\big\|_{F}^{2}\label{eq:problem-matrix-fact-unimodal}\\
\text{subject to} & \quad\bm{U}\in\mathcal{U}^{N\times K},\bm{V}\in\mathcal{U}^{N\times K}\label{eq:problem-matrix-fact-unimodal-constraint}
\end{align}
where $\varodot$ denotes the Hadamard product, i.e., $\bm{W}\varodot\bm{\mathsf{H}}$
is an $N\times N$ matrix computed entry-by-entry with $\big[\bm{W}\varodot\bm{\mathsf{H}}\big]_{ij}=W_{ij}\mathsf{H}_{ij}$.

In the remaining part of this paper, we will discuss the cases of
single source, two sources, and arbitrary number of sources, and the
corresponding methods to solve $\mathscr{P}1$ or relaxed versions
of it.

\section{Special Case I: Single Source}

\label{sec:single-source}

In the single source case, we first show that a relaxation of $\mathcal{\mathscr{P}}1$
can be easily solved by a matrix completion problem followed by \ac{svd}.
We then analyze the squared error bound of the source localization.

\subsection{Solution via Nuclear Norm Minimization}

Dropping the unimodal constraints (\ref{eq:problem-matrix-fact-unimodal-constraint})
and applying a convex relaxation on the objective function, $\mathscr{P}1$
can be relaxed to a classical rank-$K$ matrix completion problem
via nuclear norm minimization. It has been shown in the sparse signal
processing literature that, under some mild regularization conditions
on the low rank matrix $\bm{H}$, the matrix $\bm{H}$ can be recovered,
with a high probability, from the sparse and noisy observation $\bm{W}\varodot\bm{\mathsf{H}}$
\cite{CanRec:J12,CanPla:J10}. Specifically, the noisy recovery of
$\bm{H}$ can be obtained as a solution, $\hat{\bm{X}}$, to the following
convex optimization problem \cite{ChoKumNarMit:J16,CanPla:J10}: 
\begin{align}
\mathscr{P}2:\quad\underset{\bm{X}}{\text{minimize}} & \quad\|\bm{X}\|_{*}\label{eq:matrix-completion}\\
\text{subject to} & \quad\sum_{(i,j)\in\Omega}\big|X_{ij}-\mathsf{H}_{ij}\big|^{2}\leq\epsilon^{2}\nonumber 
\end{align}
where $\|\bm{X}\|_{*}$ denotes the nuclear norm of $\bm{X}$ (\emph{i.e.},
the sum of the singular values of $\bm{X}$), and $\epsilon^{2}$
is a small parameter (depending on the observation noise \cite{CanPla:J10})
for the tolerance of the observation noise in $\bm{\mathsf{\bm{H}}}$.

Note that under exact recovery $\hat{\bm{\mathsf{X}}}=\bm{H}$, the
signature vectors are exactly the dominant singular vectors of $\hat{\bm{\mathsf{X}}}$,
and the unimodal constraints (\ref{eq:problem-matrix-fact-unimodal-constraint})
are then automatically satisfied. As a result, an efficient approximate
solution to $\mathscr{P}1$ can be obtained from the dominant singular
vectors $\hat{\bm{\mathsf{u}}}_{1}$ and $\hat{\bm{\mathsf{v}}}_{1}$
of $\hat{\bm{\mathsf{X}}}$ as the solution to $\mathscr{P}2$. 

\subsection{Squared Error Bound}

\label{subsec:squared-error-bound}

Let $\bm{\mathsf{e}}_{1}^{N}=\hat{\bm{\mathsf{v}}}_{1}^{N}-\bm{v}_{1}^{N}$
be the error vector, where the superscript $N$ explicitly indicates
that the analysis is performed under an $N\times N$ grid topology.
Let $\bar{v}_{1}^{N}(x)$ be the linear interpolation function of
$\hat{\bm{\mathsf{v}}}_{1}^{N}$ attached with coordinates $\bm{c}_{\text{X}}^{N}$.
Define $\bar{e}_{1}^{N}(x)=\bar{v}_{1}^{N}(x)-w_{1}(x)$ where $w_{1}(x)\triangleq w(x-s_{1,1})$
is the limiting function defined in Proposition \ref{prop:symmetry-signature-vector}.
Define $E^{N}(s)\triangleq\int_{-\infty}^{\infty}\bar{e}_{1}^{N}(x)\bar{e}_{1}^{N}(-x+s)dx$.
Suppose that the following regularity conditions are satisfied 
\begin{equation}
\underset{N\to\infty}{\text{limsup}}\;\frac{\Big|\int_{-\infty}^{\infty}w_{1}(x)\bar{e}_{1}^{N}(-x+s)dx\Big|}{\big|(\bm{v}_{1}^{N})^{\text{T}}\bm{\mathsf{e}}_{1}^{N}\big|}\leq\frac{C_{e}}{2}<\infty\label{eq:regularity-condition}
\end{equation}
\begin{equation}
\lim_{N\to\infty}\frac{E^{N}(s)}{\big|(\bm{v}_{1}^{N})^{\text{T}}\bm{\mathsf{e}}_{1}^{N}\big|}=0\label{eq:regularity-condition-2}
\end{equation}
with probability 1, for any $-\frac{L}{2}\leq s\leq\frac{L}{2}$ and
any rotation of the coordinate system $\mathcal{C}$.
\begin{rem}
[Interpretation of (\ref{eq:regularity-condition})--(\ref{eq:regularity-condition-2})]
The value $C_{e}$ depends on the distributions of the sample locations
and sample noise, as well as the energy field function $h(d)$. An
intuitive explanation of the boundedness of $C_{e}$ and the diminishing
value of $E^{N}(s)/\big|(\bm{v}_{1}^{N})^{\text{T}}\bm{\mathsf{e}}_{1}^{N}\big|$
as $N\to\infty$ is that the observation noise is uncorrelated and
identically distributed if sampled at the opposite locations symmetric
about the source. For example, in our numerical experiment using uniformly
random sampling corrupted by Gaussian noise in an underwater acoustic
environment (see Section VII for the empirical propagation model),
properties (\ref{eq:regularity-condition})\textendash (\ref{eq:regularity-condition-2})
are observed. 
\end{rem}
In addition, let $\kappa^{N}=(\sigma_{k,1}^{N}-\sigma_{k,2}^{N})/\alpha_{1}$
be the normalized singular value gap from the model (\ref{eq:matrix-decomposition-model})
for the single source $k=1$, and suppose that the gap is non-diminishing,
\emph{i.e.}, $\kappa\triangleq\lim\inf_{N\to\infty}\kappa^{N}>0$.
Furthermore, suppose that for the autocorrelation function $\tau(t)$,
the first and the second-order derivatives, $\tau'(t)$ and $\tau''(t)$,
exist and are continuous at $t=0$. Then, for a small $\nu>0$, there
exists $a_{\nu}>0$, such that $\tau(t)\leq\tau(0)+\tau'(0)t+\frac{1}{2}(\tau''(0)+\nu)t^{2}$
for all $t\in[0,a_{\nu}]$. We have the following theorems to characterize
the asymptotic behavior of the squared estimation error $\|\hat{\bm{\mathsf{s}}}_{1}-\bm{s}_{1}\|_{2}^{2}$.

\subsubsection{Conservative construction}

Choose $N=\sqrt{M}$ and suppose that all the entries of $\bm{\mathsf{H}}$
are observed. As a result, the signature vectors $\hat{\bm{\mathsf{u}}}_{1}$
and $\hat{\bm{\mathsf{v}}}_{1}$ can be directly extracted as the
dominant left and right singular vectors of $\bm{\mathsf{H}}$. 
\begin{thm}
[Squared Error Bound under the Conservative Construction]\label{thm:squared-error-bound-full}
Under the condition of Proposition \ref{prop:symmetry-signature-vector},
for asymptotically large $M$, 
\begin{equation}
\|\hat{\bm{\mathsf{s}}}_{1}-\bm{s}_{1}\|_{2}^{2}\leq C_{1}\frac{L^{4}}{M^{1.5}}+C_{2}\frac{L^{2}}{\sqrt{M}}\frac{\bar{\sigma}_{\text{n}}^{2}}{\alpha_{1}^{2}}\label{eq:squared-error-bound-full}
\end{equation}
with probability at least $1-2e^{-C_{3}N}$, where $C_{1}=\frac{3C_{0}C_{e}K_{h}^{2}}{-\kappa^{2}(\tau''(0)+\nu)}$,
$C_{2}=\frac{6C_{0}C_{e}}{-\kappa^{2}(\tau''(0)+\nu)}$, $K_{h}$
is the Lipschitz parameter of $h(d)$, and $C_{0}$ and $C_{3}$ are
constants. 
\end{thm}
\begin{proof}
See Appendix \ref{app:pf-thm-localization-error-bound}.
\end{proof}
\begin{rem}
The constants $C_{1}$ and $C_{2}$ show that the energy decay rate
of $h(d)$ versus $d$ may play a complicated role in the localization
performance. A large decay rate, corresponding to a large Lipschitz
parameter $K_{h}$, may harm the performance, because it leads to
high discretization error. On the other hand, a small decay rate for
$h(d)$ may result in performance less tolerant to noise, because
the autocorrelation function $\tau(t)$ is less sharp. The key message
of Theorem \ref{thm:squared-error-bound-full} is to establish the
worst case performance scaling law, which will be discussed later.
\end{rem}

\subsubsection{Aggressive construction}

Choose the number of sensors $M$ and the matrix dimension $N$, such
that $M$ is the smallest integer that satisfies $M\geq CN(\log N)^{2}$.
The signature vectors $\hat{\bm{\mathsf{u}}}_{1}$ and $\hat{\bm{\mathsf{v}}}_{1}$
are extracted from $\hat{\bm{\mathsf{X}}}$, the solution to $\mathscr{P}2$.
\begin{thm}
[Squared Error Bound under the Aggressive Construction]\label{thm:squared-error-bound-partial}
Suppose that the sampling error of $\bm{\mathsf{H}}$ is bounded as
$\sum_{(i,j)}\big|\mathsf{H}_{ij}-H_{ij}\big|^{2}\leq\epsilon^{2}$,
where $\epsilon$ is the parameter used in $\mathscr{P}2$. Then,
under the condition of Proposition \ref{prop:symmetry-signature-vector},
there exists some $0<\alpha<\frac{1}{2}$, such that for asymptotically
large $M$ and high enough \ac{snr} $\alpha_{1}/\bar{\sigma}_{\text{n}}^{2}$,
\begin{equation}
\|\hat{\bm{\mathsf{s}}}_{1}-\bm{s}_{1}\|_{2}^{2}\leq C_{3}\frac{L^{4}}{M^{2-\alpha}}+C_{4}\frac{L^{3}}{M^{1-\alpha}}\frac{\bar{\sigma}_{\text{n}}}{\alpha_{1}}+C_{5}L^{2}\frac{\bar{\sigma}_{\text{n}}^{2}}{\alpha_{1}^{2}}\label{eq:squared-error-bound-partial}
\end{equation}
with probability $1-\mathcal{O}(N^{-\beta})$, where $C_{3}=\frac{C_{0}'C_{e}K_{h}^{2}}{-2\kappa^{2}(\tau''(0)+\nu)}$,
$C_{4}=\frac{\sqrt{2}C_{0}'C_{e}K_{h}}{-\kappa^{2}(\tau''(0)+\nu)}$,
$C_{5}=\frac{C_{0}'C_{e}}{-\kappa^{2}(\tau''(0)+\nu)}$, and $C_{0}'$
and $\beta$ are constants.
\end{thm}
\begin{proof}
See Appendix \ref{app:pf-thm-localization-error-bound}.
\end{proof}

\subsubsection{Discussions}

We draw the following observations.

\emph{Performance Scaling Law:} Theorems \ref{thm:squared-error-bound-full}
and \ref{thm:squared-error-bound-partial} reveal the error exponent
as the number of sensors increases. As a performance benchmark, for
a naive non-parametric peak localization scheme that estimates the
source location directly from the position of the measurement sample
that achieves the highest energy, the localization squared error decreases
as $\mathcal{O}(1/M)$, whereas, the squared error of the proposed
schemes decreases faster than $\mathcal{O}(1/M^{1.5})$ in high \ac{snr}
$\alpha_{1}/\bar{\sigma}_{\text{n}}^{2}\gg1$, order-wise faster than
the naive scheme, as will be demonstrated by our numerical results
in Section \ref{sec:numerical}. 

\emph{Sparsity and Noise Suppression Tradeoff:} While the aggressive
construction strategy achieves higher error decay rate in terms of
the number of samples $M$ under high $\mathsf{SNR}\triangleq\alpha_{1}/\bar{\sigma}_{\text{n}}^{2}\gg1$,
the aggressive construction scheme is less tolerant to measurement
noise as observed from the last terms in (\ref{eq:squared-error-bound-full})
and (\ref{eq:squared-error-bound-partial}). Specifically, in low
$\mathsf{SNR}$, the squared error bound of the aggressive construction
strategy scales as $\mathcal{O}(\frac{1}{\mathsf{SNR}})$, whereas,
it scales as $\mathcal{O}(\frac{1}{N\mathsf{SNR}})$ for the conservative
construction strategy. 

\emph{Impact from the Propagation Law:} The parameter $0<\kappa\leq1$,
as appear in the coefficients of the bound, captures how precisely
the outer product $\bm{u}_{k}\bm{v}_{k}^{\text{T}}$ of the signature
vectors may approximate the energy field matrix $\bm{H}^{(k)}$ for
a single source. In the special case when $h(d)$ is fully decomposable,
\emph{i.e.}, $h(d(x,y))=h_{1}(x)h_{2}(y)$, we have $\kappa=1$ leading
to a small error bound.\footnote{The only example for this special case is the Gaussian function, \emph{i.e.},
$h(d(x,y))=ae^{-b(x^{2}+y^{2})}=ae^{-bx^{2}}e^{-by^{2}}$.} For most practical propagation models we have tested (\emph{e.g.},
propagations of radio signals over the air, acoustic signals in the
water, etc.), $\kappa$ is close to $1$. 

\section{Special Case II: Two Sources }

\label{sec:double-source}

In the case of two sources, the \ac{svd} may not extract the desired
signature vectors from $\bm{H}$ in (\ref{eq:matrix-decomposition-model}),
because the signature vectors $\bm{u}_{1}$ and $\bm{u}_{2}$ from
different sources are not necessarily orthogonal. However, it turns
out that by choosing an appropriate coordinate system $\mathcal{C}$,
the corresponding \ac{umf} problem $\mathscr{P}1$ can be easily
solved (under some mild conditions). In this section, we propose rotation
techniques to select the best coordinate system for extracting the
source signature vectors. 

\subsection{Optimal Rotation of the Coordinate System}

Suppose that we fix the origin at the center of the target area and
rotate the coordinate system such that the two sources are aligned,
\ac{wlog}, on the $y$ axis. Then, the source locations satisfy $s_{1,1}=s_{2,1}$,
and the signature vectors follow $\bm{v}_{1}=\bm{v}_{2}$. This approximately
yields a rank-1 model $\bm{H}=\bm{H}^{(1)}+\bm{H}^{(2)}\approx\big(\alpha_{1}\bm{u}_{1}+\alpha_{2}\bm{u}_{2}\big)\bm{v}_{1}^{\text{T}}$
by ignoring the minor components in (\ref{eq:matrix-decomposition-model}).

As a result, an algorithm can be designed as follows. First, extract
the vectors $\hat{\alpha}_{1}\hat{\bm{\mathsf{u}}}_{1}+\hat{\alpha}_{2}\hat{\bm{\mathsf{u}}}_{2}$
and $\hat{\bm{\mathsf{v}}}_{1}$ by solving the matrix completion
problem $\mathscr{P}2$ followed by the \ac{svd} as developed in
Section \ref{sec:single-source}. Second, obtain $\hat{\bm{\mathsf{u}}}_{1}$
and $\hat{\bm{\mathsf{u}}}_{2}$ from the composite vector $\hat{\alpha}_{1}\hat{\bm{\mathsf{u}}}_{1}+\hat{\alpha}_{2}\hat{\bm{\mathsf{u}}}_{2}$. 

The remaining challenge is to find optimal rotated coordinate system
for source alignment from the $M$ measurement samples. Note that
the source topology is not known.

Denote $\bm{\mathsf{H}}(\theta)$ as the observation matrix constructed
in coordinate system $\mathcal{C}_{\theta}$ with $\theta$ degrees
of rotation to reference coordinate system $\mathcal{C}$. The desired
rotation $\theta$ can be obtained as 
\begin{equation}
\mathscr{P}3:\quad\underset{\theta\in[0,\frac{\pi}{2}]}{\text{maximize}}\quad\rho(\theta)\triangleq\frac{\sigma_{1}^{2}(\bm{\mathsf{H}}(\theta))}{\sum_{k=1}^{N}\sigma_{k}^{2}(\bm{\mathsf{H}}(\theta))}\label{eq:rho-function-1}
\end{equation}
where $\sigma_{k}(\bm{\mathsf{H}})$ is defined as the $k$th largest
singular value of $\hat{\bm{\mathsf{X}}}(\bm{\mathsf{H}})$, the solution
to the matrix completion problem $\mathscr{P}2$ based on $\bm{\mathsf{H}}$. 

Note that $\rho(\theta)\leq1$ for all $\theta\in[0,\frac{\pi}{2}]$.
In addition, $\rho(\theta^{*})=1$, when $\bm{\mathsf{H}}(\theta)$
becomes a rank-1 matrix where the sources are aligned with one of
the axes, which gives an intuitive justification for $\mathscr{P}3$.

However, an exhaustive search for the solution $\theta^{*}$ is computationally
expensive, since it requires performing \ac{svd} on $\bm{\mathsf{H}}(\theta)$
for each $\theta$. Yet, it can be shown that $\rho(\theta)$ has
a nice locally unimodal property that enables an efficient solution.

Let $\bm{H}(\theta)$ be the signature matrix defined in (\ref{eq:matrix-decomposition-model})
under the coordinate system $\mathcal{C}_{\theta}$. 
\begin{thm}
[Property of $\rho(\theta)$]\label{thm:Unique-local-maximum} Assume
that the two sources have equal transmission power $\alpha_{1}=\alpha_{2}$.
In addition, suppose that $\bm{H}(\theta)$ is at most rank-2 and
can be perfectly recovered from $\bm{\mathsf{H}}(\theta)$, \emph{i.e.},
$\hat{\bm{\mathsf{X}}}(\bm{\mathsf{H}}(\theta))=\bm{H}(\theta)$.
Consider that $N$ is large enough. Then, $\rho(\theta)$ is periodic,
i.e., $\rho(\theta)=\rho(\theta+\frac{\pi}{2})$. In addition, $\rho(\theta)$
is strictly increasing over $(\theta^{*}-\frac{\pi}{4},\theta^{*})$
and strictly decreasing over $(\theta^{*},\theta^{*}+\frac{\pi}{4})$,
if the energy field satisfies 
\begin{equation}
s\cdot\tau'(t)>t\cdot\tau'(s)\label{eq:correlation-condition}
\end{equation}
for all $0<s<t$, where $\tau'(t)\triangleq\frac{d}{dt}\tau(t)$ and
$\theta^{*}$ is the maximizer of $\rho(\theta)$. 
\end{thm}
\begin{proof}
See Appendix \ref{app:pf-thm-unimodal-local-maximum}. 
\end{proof}

The result in Theorem \ref{thm:Unique-local-maximum} confirms that
the function $\rho(\theta)$ has a unique local maximum within a $\frac{\pi}{2}$-window
under a mild condition, in the ideal case of perfect recovery $\hat{\bm{\mathsf{X}}}(\bm{\mathsf{H}}(\theta))=\bm{H}(\theta)$.
The property motivates a simple bisection search algorithm to efficiently
search for the globally optimal solution, $\theta^{*}$, to $\mathscr{P}3$. 

Note that condition (\ref{eq:correlation-condition}) can be satisfied
by a variety of energy fields. For example, for Laplacian field $h(x,y)=\gamma e^{-\gamma|x|-\gamma|y|}$,
we have the autocorrelation function $\tau(t)=(1+\gamma t)e^{-\gamma t}$,
and its derivative $\tau'(t)=-\gamma^{2}te^{-\gamma t}$; for Gaussian
field $h(d)=\sqrt{2\gamma/\pi}e^{-\gamma d^{2}}$, where $d^{2}=x^{2}+y^{2}$,
we have $\tau(t)=e^{-\gamma t^{2}/2}$, and $\tau'(t)=-\gamma te^{-\gamma t^{2}/2}$.
In both cases, condition (\ref{eq:correlation-condition}) is satisfied.

\subsection{Source Separation}

%

% large and small separation for detection

\Ac{wlog}, suppose that the signature matrix has a decomposition
form $\bm{H}\approx\big(\alpha_{1}\bm{u}_{1}+\alpha_{2}\bm{u}_{2}\big)\bm{v}_{1}^{\text{T}}$
after the optimal coordinate system rotation from solving $\mathscr{P}3$.
In addition, assume the equal power case $\alpha_{1}=\alpha_{2}$.
Then, using a technique similar to that in Section \ref{sec:single-source},
we can estimate the $x$ coordinates from 
\begin{equation}
\hat{\mathsf{s}}_{1,1}=\hat{\mathsf{s}}_{2,1}=\frac{1}{2}\underset{t\in\mathbb{R}}{\text{argmax}}\,R(t;\hat{\bm{\mathsf{v}}}_{1},\bm{c}_{\text{X}}).\label{eq:location-estimator-two-source-y}
\end{equation}
In addition, the $y$ coordinates can be estimated as $\hat{\mathsf{s}}_{1,2},\hat{\mathsf{s}}_{2,2}=\hat{\mathsf{c}}\pm r$,
where 
\begin{align}
\hat{\mathsf{c}}=\frac{1}{2}\underset{t\in\mathbb{R}}{\text{argmax}}\,R(t;\hat{\bm{\mathsf{u}}}_{1},\bm{c}_{\text{Y}}),\quad\hat{\mathsf{r}} & =\underset{r\geq0}{\text{argmax}}\quad Q(r;\hat{\bm{\mathsf{u}}}_{1},\hat{\bm{\mathsf{v}}}_{1})\label{eq:location-estimator-two-source-d}
\end{align}
and
\[
Q(r;\hat{\bm{\mathsf{u}}}_{1},\hat{\bm{\mathsf{v}}}_{1})\triangleq\frac{1}{2}\int_{-\infty}^{\infty}\hat{u}(x)\Big(\hat{v}(x-\hat{\mathsf{c}}-r)+\hat{v}(x-\hat{\mathsf{c}}+r)\Big)dx
\]
in which, $\hat{u}(x)$ is a linear interpolation of $\hat{\bm{\mathsf{u}}}_{1}$
(similar to $\bar{v}(x)$ in (\ref{eq:reflected-correlation})) and
$\hat{v}(x)$ is from linearly interpolating $\hat{\bm{\mathsf{v}}}_{1}$.
Using calculations similar to (\ref{eq:reflected-correaltion-w})\textendash (\ref{eq:uu-2nd-equality}),
it can be shown that $Q(r;\bm{u}_{1}+\bm{u}_{2},\bm{v}_{1})$ is maximized
at $r^{*}=(s_{1,2}-s_{2,2})/2$ if the interpolation is perfect, \emph{i.e.},
$\hat{u}(x)=w(x-s_{1,2})+w(x-s_{2,2})$ and $\hat{v}(x-s_{1,1})=w(x-s_{1,1})$.

As a remark, a simple peak finding solution may not work as well as
the estimator from (\ref{eq:location-estimator-two-source-y})\textendash (\ref{eq:location-estimator-two-source-d})
because, first, peak finding mostly depends on the measurement around
the centroid $\hat{\mathsf{c}}$ and the information from other measurements
may not be fully exploited, and second, when two sources are close
to each other, their aggregate energy field will appear as one peak,
which does not correspond to the desired source location. On the other
hand, the proposed procedure (\ref{eq:location-estimator-two-source-y})\textendash (\ref{eq:location-estimator-two-source-d})
can resolve these issues. 

\section{General Case: Arbitrary Number of Sources}

\label{sec:arbitrary-number-of-sources}

In the case of an arbitrary number of sources, we first study a general
algorithm framework to solve $\mathscr{P}1$. We then discuss efficient
approximations for fast implementation of the algorithm. Finally,
an optimization of the coordinate system $\mathcal{C}_{\theta}$ is
studied to enhance the convergence of the algorithm. 

\subsection{The Gradient Projection}

Let $f(\bm{U},\bm{V})=\big\|\bm{W}\varodot\big(\bm{\mathsf{H}}-\bm{U}\bm{V}^{\text{T}}\big)\big\|_{F}^{2}$.
With some algebra and matrix calculus, it can be shown that the gradients
of $f$ are
\begin{align*}
\frac{\partial}{\partial\bm{U}}f & =-2(\bm{W}\varodot\bm{\mathsf{H}})\bm{V}+2(\bm{W}\varodot(\bm{U}\bm{V}^{\text{T}}))\bm{V}\\
\frac{\partial}{\partial\bm{V}}f & =-2(\bm{W}^{\text{T}}\varodot\bm{\mathsf{H}}^{\text{T}})\bm{U}+2(\bm{W}^{\text{T}}\varodot(\bm{V}\bm{U}^{\text{T}}))\bm{U}
\end{align*}
and the iteration of the projected gradient algorithm can be computed
as 
\begin{align}
\bm{U}(t+1) & =\mathcal{P}_{\mathcal{U}}\Big\{\bm{U}(t)-\mu_{t}\frac{\partial}{\partial\bm{U}}f(\bm{U}(t),\bm{V}(t))\Big\}\label{eq:gradient-iteration-U}\\
\bm{V}(t+1) & =\mathcal{P}_{\mathcal{U}}\Big\{\bm{V}(t)-\nu_{t}\frac{\partial}{\partial\bm{V}}f(\bm{U}(t+1),\bm{V}(t))\Big\}\label{eq:gradient-iteration-V}
\end{align}
where $\mathcal{P}_{\mathcal{U}}\{\cdot\}$ is a projection operator
to project any $N\times K$ matrix onto the unimodal cone $\mathcal{U}^{N\times K}$,
and the step size $\mu_{t}$ and $\nu_{t}$ are chosen to ensure the
decrease of the objective function $f$ (for example, via a backtracking
line search \cite{Bertsekas:1999bs,Boyd:2004kx}). 

\subsection{Fast Unimodal Projection}

The projection $\mathcal{P}_{\mathcal{U}}\{\bm{X}\}$ onto the unimodal
cone is formally defined as the solution that minimizes $\|\bm{X}-\bm{Y}\|_{\text{F}}$
over $\bm{Y}\in\mathcal{U}^{N\times K}$. Due to the property of the
Frobenius norm, the projection can be computed column-by-column. 

While it is not straight-forward to efficiently project onto the convex
set $\mathcal{U}_{s}^{N}$ (specified by constraints (\ref{eq:unimodal-1})\textendash (\ref{eq:unimodal-2}))
as it may seem to be, it is relatively easier to compute the projection
onto an \emph{isotonic cone}, where an isotonic sequence is defined
as a non-increasing (or non-decreasing) sequence. Recently, a fast
algorithm for exact isotonic projection was developed in \cite{nemeth:J10},
which finds the solution within $N-1$ steps. With such a tool, a
fast approximate algorithm to compute $\mathcal{P}_{\mathcal{U}}\{\bm{X}\}$
can be described as follows. 

\uline{Fast unimodal projection:}

\begin{enumerate}
\item For each odd index $s$, compute the isotonic projection for $\bm{x}_{k}$,
the $k$th column of $\bm{X}$, to form an ascending branch $y_{1},y_{2},\dots,y_{s}$
and a descending branch $y_{s+1},y_{s+2},\dots,y_{N}$, respectively,
using the exact isotonic projection algorithm in \cite{nemeth:J10}.
\item Construct $\bm{y}^{(s)}:=(y_{1},y_{2,}\dots,y_{s},y_{s+1},\dots,y_{N})$,
and repeat from Step 1) to compute a series of projections $\bm{y}^{(s)}$
for $s=1,3,5,\dots$.
\item Choose the solution $\bm{y}^{(s_{*})}$ that minimizes $\|\bm{y}^{(s)}-\bm{x}_{k}\|$
over $s=1,3,5,\dots$.
\end{enumerate}
Since Step 1) has complexity at most $N$, the overall complexity
is at most $\frac{1}{2}N^{2}$. 

Note that minimizing $\|\bm{x}-\bm{y}\|_{2}^{2}$ subject to the unimodal
constraint $\bm{y}\in\mathcal{U}^{N}=\bigcup_{s=1}^{N}\mathcal{U}_{s}^{N}$
is equivalent to solving a series of minimization problems, each under
constraint $\bm{y}\in\mathcal{U}_{s}^{N}\cup\mathcal{U}_{s+1}^{N}$
for $s=1,3,5,\dots$. In addition, each subproblem is equivalent to
minimizing $\|\bm{x}_{1:s}-\bm{y}'\|_{2}^{2}+\|\bm{x}_{s+1:N}-\bm{y}''\|_{2}^{2}$
subject to $0\leq y_{1}'\leq y_{2}'\leq\dots\leq y_{s}'$ and $y_{1}''\geq y_{2}''\geq\dots\geq y_{N-s}''$,
where $\bm{x}_{a:b}=(x_{a},x_{a+1},\dots,x_{b})$ represents a vector
that is extracted from the $a$th entry to the $b$th entry of $\bm{x}$.
Therefore, the above unimodal projection is exact.

\subsection{Local Convergence Analysis}

% local convergence

We frame the analysis according to the following two observations.
First, when there is no sampling noise, the unimodal constraint $(\bm{U},\bm{V})\in\mathcal{U}^{N\times K}\times\mathcal{U}^{N\times K}$
is not active at the globally optimal point $\hat{\bm{\mathsf{X}}}=(\hat{\bm{\mathsf{U}}},\hat{\bm{\mathsf{V}}})$.
Then, we can remove the projection in (\ref{eq:gradient-iteration-U})
and (\ref{eq:gradient-iteration-V}) and analyze the local behavior
of an unconstrained algorithm approaches $\hat{\bm{\mathsf{X}}}$.
The goal is to discover any factor that possibly harms the convergence
and determines methods that will improve the performance. 

Second, note that the function is bi-convex. Therefore, we can study
partial convergence, where the convergence of the variable $\bm{U}$
is analyzed while fixing the other variable $\bm{V}$ to be in the
neighborhood of $\hat{\bm{\mathsf{V}}}$. As a result, the unconstrained
algorithm trajectory for $\bm{U}$ will converge to a unique solution.
Analyzing the asymptotic convergence rate of $\bm{U}$ may help us
understand the factors that affect algorithm convergence. 

Denote $g(\bm{X})=[\big(\partial f/\partial\bm{U})^{\text{T}}\;\big(\partial f/\partial\bm{V})^{\text{T}}]^{\text{T}}$
as the gradient function of $f(\bm{X})$, where $\bm{X}=(\bm{U},\bm{V})$.
Suppose $\bm{X}(0)$ is sufficiently close to $\hat{\bm{\mathsf{X}}}$,
such that the unimodal constraints are not active. As a continuous
counter-part to the discrete iteration (\ref{eq:gradient-iteration-U})\textendash (\ref{eq:gradient-iteration-V}),
the continuous algorithm trajectory $\bm{X}(t)$ can be given as
\begin{equation}
\frac{d}{dt}\bm{X}(t)=-g(\bm{X}(t)).\label{eq:continuous-algorithm}
\end{equation}

Let $\mathcal{E}(\bm{\mathsf{X}}_{e}(t))=\frac{1}{2}\|\bm{\mathsf{X}}_{e}(t)\|_{F}^{2}$
be the normed error function for the convergence error $\bm{\mathsf{X}}_{e}(t)\triangleq\bm{X}(t)-\hat{\bm{\mathsf{X}}}$.
Let $\bm{\mathsf{U}}_{e}=\bm{U}-\hat{\bm{\mathsf{U}}}$ and $\bm{\mathsf{V}}_{e}=\bm{V}-\hat{\bm{\mathsf{V}}}$
with the time index $t$ dropped for notational brevity. The following
result suggests that if either $\bm{\mathsf{U}}_{e}$ or $\bm{\mathsf{V}}_{e}$
is much smaller than the other variable, then the algorithm trajectory
$\bm{X}(t)$ converges exponentially to $\hat{\bm{\mathsf{X}}}$. 
\begin{prop}
[Partial convergence]\label{prop:Partial-convergence} Assume perfect
sampling $\bm{\mathsf{H}}=\bm{H}$ and $\bm{H}$ in (\ref{eq:matrix-decomposition-model})
has rank $K$. Suppose that the algorithm initialization $\bm{X}(0)$
is in the neighborhood of the optimal solution $\hat{\bm{\mathsf{X}}}$
to $\mathscr{P}1$. Then the following holds 
\begin{equation}
\frac{d}{dt}\mathcal{E}(\bm{\mathsf{X}}_{e})\leq-2\lambda_{K}(\hat{\bm{\mathsf{V}}}^{\text{T}}\hat{\bm{\mathsf{V}}})\|\bm{\mathsf{U}}_{e}\|_{\text{F}}^{2}+o\Big(\|\bm{\mathsf{U}}_{e}\|_{\text{F}}^{2}\Big)\label{eq:error-convergence-rate-U}
\end{equation}
for $\|\bm{\mathsf{V}}_{e}\|_{\text{F}}=o\big(\|\bm{\mathsf{U}}_{e}\|_{\text{F}}\big)$,
where $\lambda_{K}(\mathbf{A})$ denotes the smallest eigenvalue of
$\mathbf{A}$. Moreover, 
\begin{equation}
\frac{d}{dt}\mathcal{E}(\bm{\mathsf{X}}_{e})\leq-2\lambda_{K}(\hat{\bm{\mathsf{U}}}^{\text{T}}\hat{\bm{\mathsf{U}}})\|\bm{\mathsf{V}}_{e}\|_{\text{F}}^{2}+o\Big(\|\bm{\mathsf{V}}_{e}\|_{\text{F}}^{2}\Big)\label{eq:error-convergence-rate-V}
\end{equation}
for $\|\bm{\mathsf{U}}_{e}\|_{\text{F}}=o\big(\|\bm{\mathsf{V}}_{e}\|_{\text{F}}\big)$.
\end{prop}
\begin{proof}
See Appendix \ref{app:pf-partial-convergence}.
\end{proof}

Proposition \ref{prop:Partial-convergence} shows that the rate of
convergence depends on the eigenvalues of $\hat{\bm{\mathsf{V}}}^{\text{T}}\hat{\bm{\mathsf{V}}}$
and $\hat{\bm{\mathsf{U}}}^{\text{T}}\hat{\bm{\mathsf{U}}}$, where
$\hat{\bm{\mathsf{V}}}$ and $\hat{\bm{\mathsf{U}}}$ carry the location
signatures of the source. Specifically, if the sources are aligned
with either the $x$ axis or the $y$ axis, then either $\hat{\bm{\mathsf{U}}}$
or $\hat{\bm{\mathsf{V}}}$ tends to have identical columns, which
leads to rank deficiency of matrices $\hat{\bm{\mathsf{U}}}^{\text{T}}\hat{\bm{\mathsf{U}}}$
or $\hat{\bm{\mathsf{V}}}^{\text{T}}\hat{\bm{\mathsf{V}}}$, corresponding
to small eigenvalues $\lambda_{K}$ and hence slow convergence. This
result suggests gradient type algorithms work better when sources
are well-separated on both axes.

\subsection{Rotation for Convergence Improvement}

\label{subsec:rotation-technique}

According to Proposition \ref{prop:Partial-convergence}, we may need
to establish a coordinate system such that the sources are well separated
on both axes. However, the challenge is that we have \emph{no} prior
knowledge of the source locations. 

Recall that $\bm{\mathsf{H}}(\theta)$ denotes the observation matrix
constructed in coordinate system $\mathcal{C}_{\theta}$ with $\theta$
degrees of rotation with respect to the reference coordinate system
$\mathcal{C}$. Similar to $\mathscr{P}3$, the desired rotation $\theta$
can be obtained as 
\begin{equation}
\mathscr{P}4:\quad\underset{\theta\in[0,\frac{\pi}{2}]}{\text{minimize}}\quad\rho(\theta)\triangleq\frac{\sigma_{1}^{2}(\bm{\mathsf{H}}(\theta))}{\sum_{k=1}^{N}\sigma_{k}^{2}(\bm{\mathsf{H}}(\theta))}\label{eq:rho-function-1-1}
\end{equation}
where $\sigma_{k}(\bm{\mathsf{H}})$ is defined as the $k$th largest
singular value of $\hat{\bm{\mathsf{X}}}(\bm{\mathsf{H}})$, the solution
to the matrix completion problem $\mathscr{P}2$ based on $\bm{\mathsf{H}}$. 

While problem $\mathscr{P}3$ is to align the sources with one of
the axes, problem $\mathscr{P}4$ tries to avoid alignments with any
axes. 

% Numerical results here

\begin{figure}
\begin{centering}
\psfragscanon
\psfrag{Rooted MSE}[][][0.7]{Root MSE [km]}
\psfrag{number of sensors, M}[][][0.7]{Number of sensors, $M$}
\psfrag{0}[][][0.7]{0}
\psfrag{20}[][][0.7]{20}
\psfrag{40}[][][0.7]{40}
\psfrag{60}[][][0.7]{60}
\psfrag{80}[][][0.7]{80}
\psfrag{120}[][][0.7]{120}
\psfrag{140}[][][0.7]{140}
\psfrag{0.1}[][][0.7]{0.1}
\psfrag{0.2}[][][0.7]{0.2}
\psfrag{0.3}[][][0.7]{0.3}
\psfrag{100}[][][0.7]{100}
\psfrag{Simple UMF}[Bl][Bl][0.65]{Simple UMF}
\psfrag{Completed UMF}[Bl][Bl][0.65]{Complete UMF}
\psfrag{Rotated UMF}[Bl][Bl][0.65]{Rotated UMF}
%\psfrag{Naive factoriz}[Bl][Bl][0.65]{Naive factoriz}
%\psfrag{Nonlinear regress}[Bl][Bl][0.65]{Kernel regress \cite{NguJorSin:J05}}
\includegraphics[width=1\columnwidth]{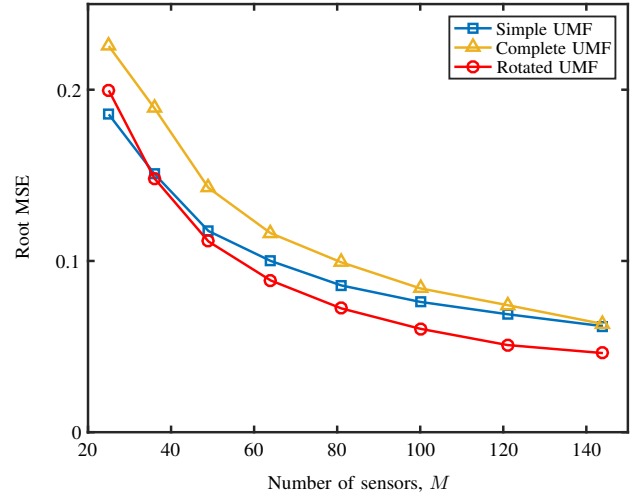}
\par\end{centering}
\caption{\label{fig:map} Localization performance under noise-free sampling
in a two source Gaussian energy field.}
\end{figure}
Fig.~\ref{fig:map} demonstrates the performance of the \ac{umf}
with optimal coordinate system rotation under noise-free sampling
$\sigma_{\text{n}}^{2}=0$, where $L=1$ and the energy field is given
by $h(d)=\sqrt{2\gamma/\pi}\exp(-\gamma d^{2})$ with $\lambda=20$.
The observation matrices constructed with dimension $N$ satisfy $N^{2}/2\approx M$.
There are two key observations: (i) the coordination system rotation
does improve the convergence as demonstrated by the comparison between
scheme ``Rotated UMF'', which solves $\mathscr{P}1$ in the optimal
coordinate system $\mathcal{C}_{\theta^{*}}$ with $\theta^{*}$ solved
from $\mathscr{P}4$, and scheme ``Simple UMF'', which solves $\mathscr{P}1$
in a fixed coordinate system $\mathcal{C}$. (ii) \Ac{umf} performs
better in the recovery of sparse unimodal structures as compared to
conventional sparse matrix completion methods, scheme ``Complete
UMF'', which first solves the matrix completion problem $\mathscr{P}2$
in the optimal coordinate system $\mathcal{C}_{\theta^{*}}$, and
then solves $\mathscr{P}1$ for signature vector extraction. As the
\ac{umf} is aware of the unimodal structure, the proposed ``Rotated
UMF'' outperforms the ``Complete UMF'' scheme.\footnote{However, we also note that in the case of low \ac{snr}, \emph{i.e.},
$\sigma_{\text{n}}^{2}\gg1$, $\mathscr{P}2$ can serve as an efficient
de-noising step to help faster convergence of \ac{umf}. }

\section{Numerical Results}

\label{sec:numerical}

\begin{figure}
\begin{centering}
{
\psfragscanon
\psfrag{MSE}[][][0.7]{Root MSE [km]}
\psfrag{M}[][][0.7]{Number of samples, $M$}
%\psfrag{1}[][][0.7]{1}
\psfrag{0.1}[][][0.7]{0.1}
\psfrag{0.01}[][][0.7]{0.01}
\psfrag{10}[][][0.7]{10}
\psfrag{100}[][][0.7]{100}
\psfrag{1000}[][][0.7]{1000}
\psfrag{Naive scheme}[Bl][Bl][0.7]{Naive scheme}
\psfrag{Kernel method}[Bl][Bl][0.7]{Kernel method}
\psfrag{Proposed}[Bl][Bl][0.7]{Proposed}
\psfrag{WCL, Square, M/4}[Bl][Bl][0.6]{WCL, $M_\text{r}=M/4$}
\psfrag{WCL, Square, M/2}[Bl][Bl][0.6]{WCL, $M_\text{r}=M/2$}
\psfrag{Decay rate}[Bl][Bl][0.65]{Decay rate}\includegraphics[width=1\columnwidth]{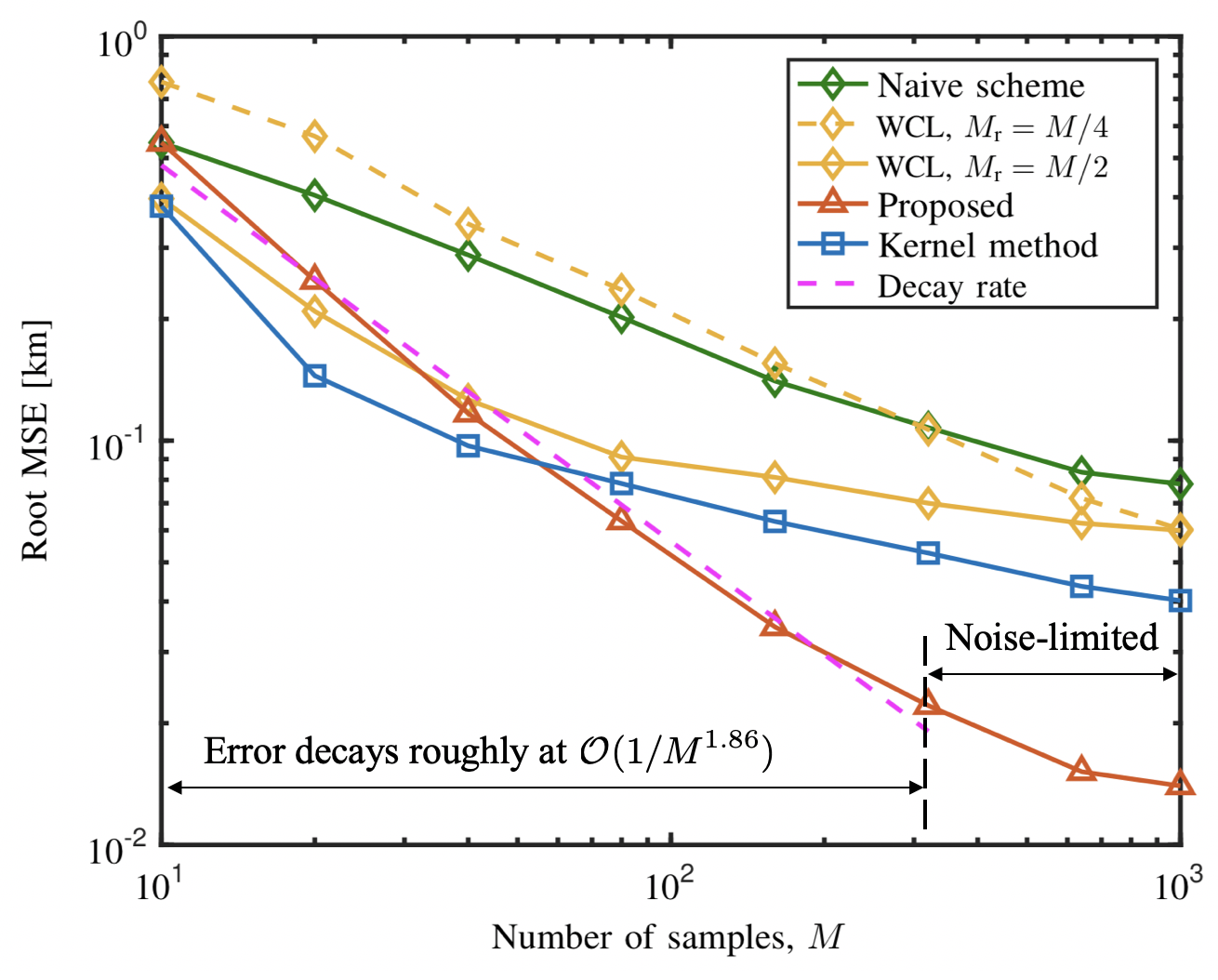}}
\par\end{centering}
\caption{\label{fig:mse-1} \ac{mse} of the source location versus the number
of sensors, $M$, in the single source case.}
\end{figure}
We consider the source and sensor deployment model in Section \ref{sec:system-model}
with $L=5$ km in an underwater environment. The sources simultaneously
transmit signals at $f=5$ kHz. The propagation of the signal from
each source is modeled using $N_{\text{p}}=15$ discrete paths, where
the path inter-arrival times $\tau_{p+1}-\tau_{p}$ are exponentially
distributed with mean $1$00 ms, and the path amplitudes are Rayleigh
distributed with power scaled as $e^{-\varphi(\tau_{p}-\tau_{1})}$
with respect to the first path, $\varphi=2$~$\text{sec}^{-1}$ \cite{beygi2015multi,brekhovskikh2003fundamentals}.
The path energy attenuation is modeled as $(1+d^{1.5}A(f)^{d})^{-1}$,
where Thorp's formula \cite{brekhovskikh2003fundamentals} is used
to arrive at $10\log_{10}A(f)=0.11f^{2}/(1+f^{2})+44f^{2}/(4100+f^{2})+2.75\times10^{-4}f^{2}+0.003$
dB/km. The ambient noise $\mathsf{n}^{(m)}$ is modeled as a zero
mean, Gaussian random variable with normalized variance $\sigma_{\text{n}}^{2}/P=-34$
dB, where $P=1$ is the total transmission power. The sensor has a
receive window of 4 seconds from the detection of the first path.
The parameter $N$ for constructing the observation matrix $\bm{\mathsf{H}}$
is chosen as the largest integer satisfying $1.5N(\log N)^{2}\leq M$. 

The proposed algorithms for the single source case ((\ref{eq:location-estimator})
in Section \ref{sec:single-source}), two source ((\ref{eq:location-estimator-two-source-y})
\textendash{} (\ref{eq:location-estimator-two-source-d}) in Section
\ref{sec:double-source}), and an arbitrary number of sources (Section
\ref{sec:arbitrary-number-of-sources}) is evaluated. Specifically,
the projected gradient algorithm in (\ref{eq:gradient-iteration-U})\textendash (\ref{eq:gradient-iteration-V})
are initialized with random vectors (entries independently and uniformly
distributed in $[0,1]$) projected on the unimodal cone. The algorithm
was run 5 times, each with independent initializations and a maximum
of 200 iterations. The solution (out of the 5) that yields the smallest
objective function value was selected.

Three baseline schemes are evaluated. Baseline 1, \emph{Naive scheme
for single source}: the location of the sensor that observes the highest
energy is identified as the source location. Baseline 2, \emph{Weighted
centroid localization} \cite{wang2011weighted}: the location estimate
is updated by 
\[
\hat{\bm{\mathsf{s}}}(n+1)=\frac{\sum_{m\in\mathcal{R}(\hat{\bm{\mathsf{s}}}(n))}\varrho^{(m)}\bm{z}^{(m)}}{\sum_{m\in\mathcal{R}(\hat{\bm{\mathsf{s}}}(n))}\varrho^{(m)}}
\]
until convergence, where $\varrho^{(m)}=(\mathsf{h}^{(m)})^{2}$ are
the squared-weights and $\mathcal{R}(\bm{x})$ specifies a set of
measurements that are taken within a radius $r$ from location $\bm{x}$,
in which, the radius $r$ is at least $L/8$ or as large as to include
exactly $M_{\text{r}}\in\{M/4,M/2\}$ measurements. Baseline 3, \emph{Kernel
regression}: The algorithm chooses parameters $\{\alpha_{k},\hat{\bm{z}}_{k}\}$
and $\lambda$ to minimize $\sum_{m}\big|\mathsf{h}^{(m)}-\sum_{k=1}^{K}\alpha_{k}B(\bm{z}^{(m)},\hat{\bm{z}}_{k};\lambda)\big|^{2}$
based on the measurement $\{\mathsf{h}^{(m)},\bm{z}^{(m)}\}$, where
two classes of kernel functions are considered, Gaussian kernel $B_{\text{G}}(\bm{z},\hat{\bm{z}}_{k};\lambda)=\exp(-\lambda\|\bm{z}-\hat{\bm{z}}_{k}\|^{2})$
and Laplacian kernel $B_{\text{L}}(\bm{z},\hat{\bm{z}}_{k};\lambda)=\exp(-\lambda\|\bm{z}-\hat{\bm{z}}_{k}\|_{1})$,
in which, $\|\cdot\|_{1}$ denotes the $L_{1}$ norm. Such a least-squares
problem is solved 5 times with randomized initializations. Cross-validation
is used to choose the best kernel function, and the data set is partitioned
to 70\% for parameter training and 30\% for \ac{mse} performance
validation. The parameters $\hat{\bm{z}}_{k}$ give the location estimates.
Note that Baselines 1 and 2 cannot differentiate multiple sources
when the two sources are close to each other. Therefore, these two
baselines are evaluated in the single source case only. 

Fig.~\ref{fig:mse-1} shows the \ac{rmse} of the source location
versus the number of sensors, $M$, in the single source case. The
proposed scheme outperforms the naive scheme and the weighted centroid
scheme.\footnote{It is observed from our numerical experiments that the weights $\varrho$
and the window $\mathcal{R}$ in the weighted centroid scheme should
depend on $M$ and $h(d)$. However, choosing the best weights and
window are highly non-trivial.} It also outperforms the kernel method from moderate to large number
of sensors. The dashed line with slope $-0.93$ shows that the error
decay rate of the proposed algorithm in the $M=10$ to $500$ region
is roughly $\mathcal{O}(1/M^{1.86})$. Note that Theorem \ref{thm:squared-error-bound-full}
and \ref{thm:squared-error-bound-partial} show that the \ac{mse}
of the proposed scheme should decrease faster than $\mathcal{O}(1/M^{1.5})$,
and hence is confirmed by our numerical results. In the case of more
than $500$ sensors, the error decay rate is limited by the ambient
noise. This phenomenon also matches with the second and third terms
of (\ref{eq:squared-error-bound-partial}) in Theorem \ref{thm:squared-error-bound-partial}.

\begin{figure}
\begin{centering}
{
\psfragscanon
\psfrag{Rooted MSE}[][][0.7]{Root MSE [km]}
\psfrag{M sensors}[][][0.7]{Number of sensors, $M$}
\psfrag{1}[][][0.7]{1}
\psfrag{0}[][][0.7]{0}
\psfrag{0.2}[][][0.7]{0.2}
\psfrag{0.4}[][][0.7]{0.4}
\psfrag{0.6}[][][0.7]{0.6}
\psfrag{0.8}[][][0.7]{0.8}
\psfrag{10}[][][0.7]{10}
\psfrag{100}[][][0.7]{100}
\psfrag{1000}[][][0.7]{1000}
\psfrag{Kernel method, K = 3XX}[Bl][Bl][0.62]{Kernel method, $K= 3$}
\psfrag{Kernel method, K = 2}[Bl][Bl][0.62]{Kernel method, $K= 2$}
\psfrag{Proposed, K = 2}[Bl][Bl][0.62]{Proposed, $K=2$}
\psfrag{Proposed, K = 3}[Bl][Bl][0.62]{Proposed, $K=3$}
\psfrag{Proposed, K = 4}[Bl][Bl][0.62]{Proposed, $K=4$}\includegraphics[width=1\columnwidth]{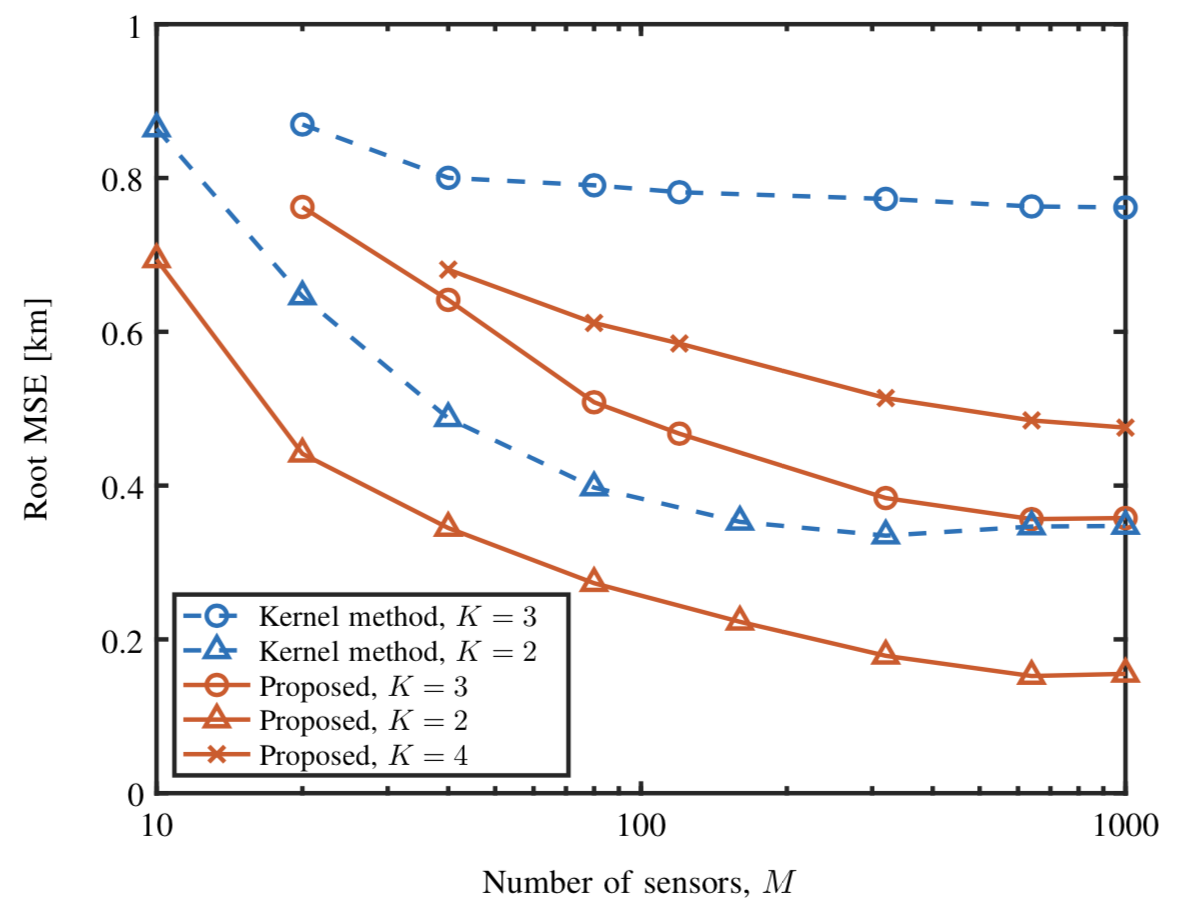}}
\par\end{centering}
\caption{\label{fig:mse-2-3} \ac{mse} of the source location versus the number
of sensors, $M$, in two to four source cases. }
\end{figure}
Fig.~\ref{fig:mse-2-3} shows the \ac{rmse} of the source locations
versus the number of sensors, $M$, in two to four source cases. The
proposed methods significantly outperform the kernel methods. Note
that the proposed schemes only require the generic property that the
source energy field $h(d)$ is non-negative and strictly decreasing
in the distance $d$ to the source. Although the kernel functions
also capture such property, the kernel methods suffer from parameter
estimation error when fitting the data to inaccurate parametric models.

\section{Conclusions}

\label{sec:conclusion}

This paper developed non-parametric algorithms for localizing multiple
sources based on a moderate number of energy measurements from different
locations. A matrix observation model was proposed and it is proven
that the dominant singular vectors of the matrix for a single source
are unimodal and symmetric. A non-parametric source localization algorithm
exploiting the unimodal and symmetry property was developed and shown
to decrease the localization \ac{mse} faster than $\mathcal{O}(1/M^{1.5})$
using $M$ sensors in the single source case with noiseless measurements.
In the two source case, the source locations can be found by choosing
the optimal rotation of the coordinate system such that the normalized
dominant singular value of the sample matrix is maximized. In the
case of arbitrary number of sources, we localize the sources by solving
a \ac{umf} problem in an optimally rotated coordinate system. Our
numerical experiments demonstrate that the proposed scheme achieves
similar performance as the baselines using no more than $1/5$ measurement
samples. 

\appendices
%

% Derivation/proof starts from here

\section{Proof of Theorem \ref{thm:unimodal-signature-vector}}

\label{app:pf-unimodal-sinature-vector}

As the signature vectors $\bm{u}_{k}$ and $\bm{v}_{k}$ correspond
to the dominant singular vectors $\bm{u}_{k,1}$ and $\bm{v}_{k,1}$
of $\bm{H}^{(k)}$, we can focus on only one source and drop the source
index $k$ here for brevity. Specifically, from (\ref{eq:matrix-decomposition-model}),
we write $\bm{H}=\alpha\bm{u}\bm{v}^{\text{T}}+\sum_{i=2}^{N}\lambda_{i}\bm{u}_{i}\bm{v}_{i}^{\text{T}}$
(for the $k$th source, where $k=1$), in which $\bm{u}$ and $\bm{u}_{i}$
are the left singular vectors of $\bm{H}$, $\bm{v}$ and $\bm{v}_{i}$
are the right singular vectors, and $\alpha$ is the largest singular
value, $\alpha>\lambda_{i}$, $i=2,3,\dots,N$. 

Let $\bm{R}=\bm{H}^{\text{T}}\bm{H}$. Then, the $(i,j)$th entry
of $\bm{R}$ is given by $R_{ij}=\bm{h}_{i}^{\text{T}}\bm{h}_{j}$,
where $\bm{h}_{i}$ is the $i$th column of $\bm{H}$. In the following
lemma, we show that the columns of $\bm{R}$ are unimodal.
\begin{lem}
\label{lem:pf-unimodal} Suppose that the source is located at the
$(m,n)$th grid cell centered at $\bm{c}_{m,n}$. Then, for each column
of $\bm{R}$, the entries $R_{ij}$ are increasing, $R_{ij}<R_{i+1,j}$,
if $i<n$, and they are decreasing, $R_{ij}>R_{i+1,j}$, if $i\geq n$. 
\end{lem}
\begin{proof}
Since the source location $\bm{s}$ is inside the $(m,n)$th grid
cell centered at $\bm{c}_{m,n}$, we have $d(\bm{c}_{p,i},\bm{s})>d(\bm{c}_{p,i+1},\bm{s})\geq d(\bm{c}_{p,n},\bm{s})$,
for $i<n$ and all $p=1,2,\dots,N$. Similarly, $d(\bm{c}_{p,i},\bm{s})<d(\bm{c}_{p,i+1},\bm{s})$,
for $i\geq n$ and all $p$. Recall that $H_{ij}=h(d(\bm{c}_{i,j},\bm{s}))$
and $h(d)$ is a non-negative decreasing function. We thus have
\begin{align*}
R_{ij}=\bm{h}_{i}^{\text{T}}\bm{h}_{j} & =\sum_{p=1}^{N}H_{pi}H_{pj}\\
 & <\sum_{p=1}^{N}H_{p,i+1}H_{pj}=\bm{h}_{i+1}^{\text{T}}\bm{h}_{j}=R_{i+1,j}
\end{align*}
for $i<n$. Similarly, we can show that $R_{ij}=\bm{h}_{i}^{\text{T}}\bm{h}_{j}>\bm{h}_{i+1}^{\text{T}}\bm{h}_{j}=R_{i+1,j}$,
for $i\geq n$. 
\end{proof}

Under the condition of Lemma \ref{lem:pf-unimodal}, if we raise $\bm{R}$
to the power of $q$, the columns of $\bm{R}^{q}$ are unimodal, with
their peaks at the $n$th entry. Specifically, define $\bm{R}^{(q)}\triangleq\bm{R}^{q}/\text{tr}\big\{\bm{R}^{q}\big\}$.
We show, in the following lemma, that the columns of $\bm{R}^{(q)}$
are unimodal. 
\begin{lem}
\label{lem:pf-unimodal-p} Let $R_{ij}^{(q)}$ be the $(i,j)$th entry
of $\bm{R}^{(q)}$. Then, $R_{ij}^{(q)}<R_{i+1,j}^{(q)}$, if $i<n$,
and $R_{ij}^{(q)}>R_{i+1,j}^{(q)}$, if $i\geq n$. 
\end{lem}
\begin{proof}
First, it can be easily verified that the result holds for $q=1$
according to Lemma \ref{lem:pf-unimodal}. Then, suppose that the
result holds for some $q\geq1$. Note that $\bm{R}^{(q+1)}=\text{tr}\big\{\bm{R}^{(q)}\big\}\bm{R}^{(q)}\bm{R}/\text{tr}\big\{\bm{R}^{(q+1)}\big\}$
and that $\bm{R}^{(q)}$ is symmetric. We have 
\begin{align*}
R_{ij}^{(q+1)} & =\frac{\text{tr}\big\{\bm{R}^{q}\big\}}{\text{tr}\big\{\bm{R}^{q+1}\big\}}\sum_{p=1}^{N}R_{pi}^{(q)}R_{pj}\\
 & <\frac{\text{tr}\big\{\bm{R}^{q}\big\}}{\text{tr}\big\{\bm{R}^{q+1}\big\}}\sum_{p=1}^{N}R_{p,i+1}^{(q)}R_{pj}=R_{i+1,j}^{(q+1)}
\end{align*}
for $i<n$. Similarly, we can show that $R_{ij}^{(q+1)}>R_{i+1,j}^{(q+1)}$,
for $i\geq n$. Therefore, by deduction, the result holds for all
$q\geq1$. 
\end{proof}

Let $\bm{R}^{\infty}\triangleq\lim_{q\to\infty}\bm{R}^{q}/\text{tr}\big\{\bm{R}^{q}\big\}$.
It can be shown that the limit exists and equals $\bm{R}^{\infty}=\bm{v}\bm{v}^{\text{T}}$.
To see this, we can easily compute $\bm{R}^{q}=\alpha^{2q}\bm{v}\bm{v}^{\text{T}}+\sum_{i=2}^{N}\lambda_{i}^{2q}\bm{v}_{i}\bm{v}_{i}^{\text{T}}$
and the normalization term $\text{tr}\big\{\bm{R}^{q}\big\}=(\alpha^{2q}+\sum_{i=2}^{N}\lambda_{i}^{2q})$.
By the Perron-Frobenius theorem, the dominant eigenvalue of a non-negative
matrix has multiplicity $1$, \emph{i.e.}, $\alpha=\lambda_{1}>|\lambda_{i}|$
for all $i\geq2$. Hence, $(\lambda_{i}/\alpha)^{2q}\to0$ as $q\to\infty$,
for $i=2,3,\dots,N$, leading to $\bm{R}^{\infty}$ being rank-1. 

On the other hand, from Lemma \ref{lem:pf-unimodal-p}, each column
of $\bm{R}^{\infty}$ is unimodal. Note that the $i$th column of
$\bm{R}^{\infty}$ can be written as $v_{i}\bm{v}$, where $v_{i}$
is the $i$th entry of $\bm{v}$. We therefore confirm that $\bm{v}$
is unimodal, with its $n$th entry being the peak.

Similarly, by constructing $\bm{Q}=\bm{H}\bm{H}^{\text{T}}$, we can
also show that $\bm{u}$ is unimodal with its $m$th entry being the
peak. \hfill\IEEEQED

\section{Proof of Proposition \ref{prop:symmetry-signature-vector}}

\label{app:pf-homogenous}

Consider uniformly shifting the discretization grid points $\{\bm{c}_{i,j}\}$
to such a position that the source location $\bm{s}_{k}$ is the center
of a rectangle formed by the four nearest grid points. Specifically,
the new grid points are given by $\tilde{\bm{c}}_{i,j}=\bm{c}_{i,j}+\bm{\delta}$,
where $\bm{\delta}=\bm{s}_{k}-\bm{c}_{m,n}-(\frac{L}{2N},\frac{L}{2N})$,
$m=\sup\{i:c_{\text{X},i}\leq s_{k,1}\}$, and $n=\sup\{j:c_{\text{Y},j}\leq s_{k,2}\}$.
One can verify that $\bm{s}_{k}$ is at the center of the rectangle
formed by the four nearest grid points $\tilde{\bm{c}}_{m,n}$, $\tilde{\bm{c}}_{m+1,n},$$\tilde{\bm{c}}_{m,n+1}$,
and $\tilde{\bm{c}}_{m+1,n+1}$.

\subsection{The Symmetry Property}

Let $\tilde{\bm{H}}^{(k)}$ be the \emph{virtual} signature matrix
defined on the shifted grid points $\{\tilde{\bm{c}}_{i,j}\}$ using
(\ref{eq:matrix-observation-model-true}). As illustrated in Fig.
\ref{fig:h-block}, consider a sub-block $\tilde{\bm{H}}_{1}^{(k)}$
that contains the $(m+1-J_{1})$th to $(m+J_{1})$th rows and the
$(n+1-J_{2})$th to $(n+J_{2})$th columns of $\tilde{\bm{H}}^{(k)}$,
where $J_{1}=\min\{m,\frac{N}{2}\}$ and $J_{2}=\min\{n,\frac{N}{2}\}$.
Further divide the $2J_{1}\times2J_{2}$ sub-block $\tilde{\bm{H}}_{1}^{(k)}$
into four $J_{1}\times J_{2}$ sub-blocks. Since the locations $\tilde{\bm{c}}_{i,j}$
that correspond to the entries of $\tilde{\bm{H}}_{1}^{(k)}$ are
symmetric about the source $\bm{s}_{k}$, one can verify that $\tilde{\bm{H}}_{1}^{(k)}$
has the following structure 
\begin{equation}
\tilde{\bm{H}}_{1}^{(k)}=\left[\begin{array}{cc}
\tilde{\bm{H}}_{11}^{(k)} & \Pi_{\text{c}}(\tilde{\bm{H}}_{11}^{(k)})\\
\Pi_{\text{r}}(\tilde{\bm{H}}_{11}^{(k)}) & \Pi_{\text{r}}\big(\Pi_{\text{c}}(\tilde{\bm{H}}_{11}^{(k)})\big)
\end{array}\right]\label{eq:H1}
\end{equation}
where the operators $\Pi_{\text{c}}(\bm{A})$ reverses the columns
of $\bm{A}$ and $\Pi_{\text{r}}(\bm{A})$ reverses the rows of $\bm{A}$. 

\begin{figure}
\begin{centering}
\includegraphics{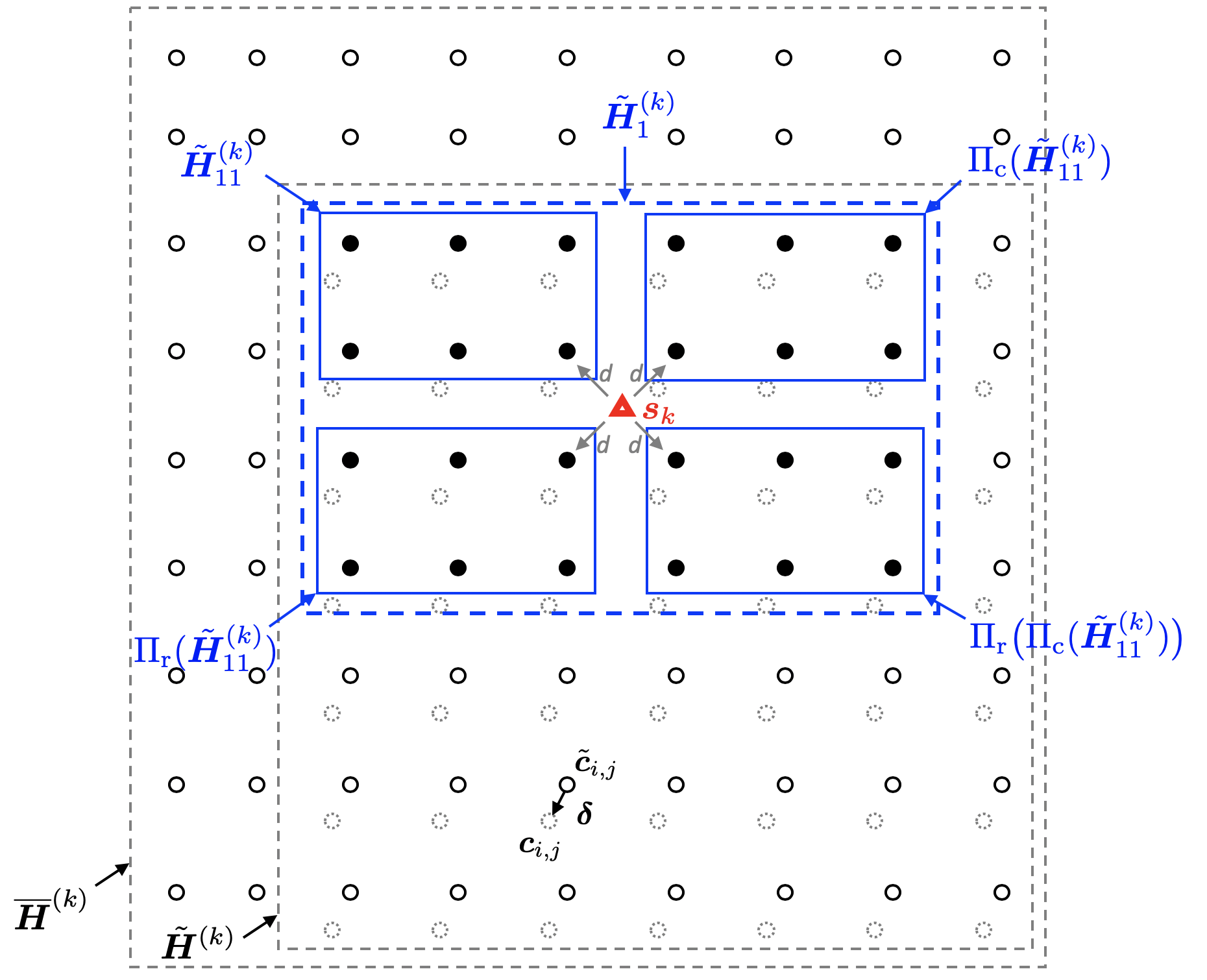}
\par\end{centering}
\caption{\label{fig:h-block} Construction of sub-blocks of $\tilde{\bm{H}}^{(k)}$
for the proof of Proposition \ref{prop:symmetry-signature-vector}.}
\end{figure}
\begin{lem}
Let $\tilde{\bm{H}}_{11}^{(k)}=\bm{U}\bm{\Sigma}\bm{V}^{\text{T}}$
be the \ac{svd} of $\tilde{\bm{H}}_{11}^{(k)}$. Then, 
\begin{equation}
\tilde{\bm{H}}_{1}^{(k)}=\left[\begin{array}{c}
\bm{U}\\
\Pi_{\text{r}}(\bm{U})
\end{array}\right]\bm{\Sigma}\left[\begin{array}{cc}
\bm{V}^{\text{T}} & \Pi_{\text{r}}(\bm{V})^{\text{T}}\end{array}\right]\label{eq:app-pf-symmetric-svd}
\end{equation}
where, on the right hand side, the left and the right matrices are
semi-orthogonal, respectively. 
\end{lem}
\begin{proof}
First, one can verify that the left-bottom sub-block of $\tilde{\bm{H}}_{1}^{(k)}$
in (\ref{eq:H1}) satisfies $\Pi_{\text{r}}(\tilde{\bm{H}}_{11}^{(k)})=\Pi_{\text{r}}(\bm{U})\bm{\Sigma}\bm{V}^{\text{T}}$,
since the $(i,j)$th entry is given by
\begin{align*}
\Big[\Pi_{\text{r}}(\tilde{\bm{H}}_{11}^{(k)})\Big]_{i,j} & =\Big[\tilde{\bm{H}}_{11}^{(k)}\Big]_{J_{1}+1-i,j}\\
 & =\sum_{l=1}^{\min\{J_{1},J_{2}\}}U_{J_{1}+1-i,l}\Sigma_{l,l}V_{j,l}\\
 & =\sum_{l=1}^{\min\{J_{1},J_{2}\}}\big[\Pi_{\text{r}}(\bm{U})\big]_{i,l}\Sigma_{l,l}V_{j,l}\\
 & =\big[\Pi_{\text{r}}(\bm{U})\bm{\Sigma}\bm{V}^{\text{T}}\big]_{i,j}.
\end{align*}
Similarly, one can verify that other sub-blocks of $\tilde{\bm{H}}_{1}^{(k)}$
agree with the decomposition (\ref{eq:app-pf-symmetric-svd}). 

Second, to see the semi-orthogonality, one can compute 
\begin{align*}
\left[\begin{array}{c}
\bm{U}\\
\Pi_{\text{r}}(\bm{U})
\end{array}\right]^{\text{T}}\left[\begin{array}{c}
\bm{U}\\
\Pi_{\text{r}}(\bm{U})
\end{array}\right] & =\bm{U}^{\text{T}}\bm{U}+\Pi_{\text{r}}(\bm{U})^{\text{T}}\Pi_{\text{r}}(\bm{U})\\
 & =2\bm{U}^{\text{T}}\bm{U}=2\bm{I}.
\end{align*}
Similarly, one can verify the semi-orthogonality of $\big[\begin{array}{cc}
\bm{V}^{\text{T}} & \Pi_{\text{r}}(\bm{V})^{\text{T}}\end{array}\big]^{\text{T}}$.
\end{proof}

As a result, the dominant singular vectors of $\tilde{\bm{H}}_{1}^{(k)}$
are $\tilde{\bm{u}}_{k}^{\text{I}}=\frac{1}{\sqrt{2}}[\bm{u}_{1}^{\text{T}},\;\Pi_{\text{r}}(\bm{u}_{1})^{\text{T}}]^{\text{T}}$
and $\tilde{\bm{v}}_{k}^{\text{I}}=\frac{1}{\sqrt{2}}[\bm{v}_{1}^{\text{T}},\;\Pi_{\text{r}}(\bm{v}_{1})^{\text{T}}]^{\text{T}}$,
where $\bm{u}_{1}$ and $\bm{v}_{1}$ are the first columns of $\bm{U}$
and $\bm{V}$, respectively. It is clear that $\tilde{\bm{u}}_{k}^{\text{I}}$
and $\tilde{\bm{v}}_{k}^{\text{I}}$ are symmetric, respectively. 

\subsection{The Zero-padding Property}

Consider building an unbounded lattice by extending the rows and columns
of $\{\tilde{\bm{c}}_{i,j}:i,j=1,2,\dots,N\}$ to infinity with equal
spacing $\frac{L}{N}$ (\emph{i.e.}, $i,j$ take all the integer values).
 We construct a $2J\times2J$ matrix $\overline{\bm{H}}^{(k)}$,
$J=N-\min\{m,n\}$, using the $2J\times2J$ array of elements $\{\frac{L}{N}\alpha_{k}h(d(\tilde{\bm{c}}_{i,j},\bm{s}_{k})):1-J\leq i-m,j-n\leq J\}$.
Then, $\tilde{\bm{H}}_{1}^{(k)}$ is a sub-block of $\tilde{\bm{H}}^{(k)}$
which is a sub-block of $\overline{\bm{H}}^{(k)}$. Moreover, $\tilde{\bm{H}}_{1}^{(k)}$
is a sub-block at the center of $\overline{\bm{H}}^{(k)}$.
\begin{lem}
Except sub-block $\tilde{\bm{H}}_{1}^{(k)}$, the entries of $\overline{\bm{H}}^{(k)}$
are zero.
\end{lem}
\begin{proof}
By the choices of $J_{1}$ and $J_{2}$, the sub-block $\tilde{\bm{H}}_{1}^{(k)}$
covers the whole area that observes non-zero energy from the source.
To see this, for any location $\tilde{\bm{c}}_{i,j}$ outside sub-block
$\tilde{\bm{H}}_{1}^{(k)}$, either $\tilde{\bm{c}}_{i,j}$ or $\tilde{\bm{c}}_{\tilde{i},\tilde{j}}$,
the grid point symmetric about the source location $\bm{s}_{k}$,
is outside the target area $\mathcal{A}$. We thus have $h(d(\tilde{\bm{c}}_{i,j},\bm{s}_{k}))=h(d(\tilde{\bm{c}}_{\tilde{i},\tilde{j}},\bm{s}_{k}))=0$
due to the assumption that $\int_{\mathbb{R}^{2}\backslash\mathcal{A}}h(d(\bm{z},\bm{s}))^{2}d\bm{z}=0$. 
\end{proof}

Thus, the singular vectors of $\overline{\bm{H}}^{(k)}$ can be obtained
from the singular vectors of $\tilde{\bm{H}}_{1}^{(k)}$ by padding
zeros with proper dimensions: $\tilde{\bm{u}}_{k}^{\text{II}}=[\bm{0}^{\text{T}},\,(\tilde{\bm{u}}_{k}^{\text{I}})^{\text{T}},\,\bm{0}^{\text{T}}]^{\text{T}}$
and $\tilde{\bm{v}}_{k}^{\text{II}}=[\bm{0}^{\text{T}},\,(\tilde{\bm{v}}_{k}^{\text{I}})^{\text{T}},\,\bm{0}^{\text{T}}]^{\text{T}}$.

In addition, $\overline{\bm{H}}^{(k)}$ is symmetric because the $(\tilde{i},\tilde{j})$th
entry is sampled at location $\tilde{\bm{c}}_{m-J+\tilde{i},n-J+\tilde{j}}$
which has the same distance to the source $\bm{s}_{k}$ as location
$\tilde{\bm{c}}_{m-J+\tilde{j},n-J+\tilde{i}}$ does. As a result,
the dominant left and right singular vectors are equal, $\tilde{\bm{u}}_{k}^{\text{II}}=\tilde{\bm{v}}_{k}^{\text{II}}$.
Since $\tilde{\bm{H}}_{1}^{(k)}$ is a sub-block of $\tilde{\bm{H}}^{(k)}$
which is a sub-block of $\overline{\bm{H}}^{(k)}$, the dominant singular
vectors of $\tilde{\bm{H}}^{(k)}$ equals to some portions of the
common vector $\tilde{\bm{u}}_{k}^{\text{II}}$, where the left singular
vector $\tilde{\bm{u}}_{k}$ is symmetric about the $n$th element,
and the right singular vector $\tilde{\bm{v}}_{k}$ the $m$th element.

\subsection{Convergence Analysis}

Finally, we study the perturbation or error $\tilde{\bm{H}}^{(k)}-\bm{H}^{(k)}$.
Denote $\tilde{u}_{k}^{N}(y)$ and $\tilde{v}_{k}^{N}(x)$ as linearly
interpolated functions from the singular vectors $\tilde{\bm{u}}_{k}$
and $\tilde{\bm{v}}_{k}$, respectively, \emph{i.e.}, $\tilde{u}_{k}^{N}(c_{\text{Y},i}+\delta_{1})=\sqrt{N/L}\tilde{u}_{k,i}$
and $\tilde{v}_{k}^{N}(c_{\text{X},j}+\delta_{2})=\sqrt{N/L}\tilde{v}_{k,j}$.
A perturbation analysis yields 
\begin{align}
\big|H_{ij}^{(k)}-\tilde{H}_{ij}^{(k)}\big| & =\frac{L}{N}\alpha_{k}\bigg|h(d_{ij})\nonumber \\
 & \qquad\quad\quad-\Big(h(d_{ij})+\int_{0}^{1}h'(d_{ij}+t\tilde{\delta}_{ij})dt\Big)\bigg|\nonumber \\
 & \leq\frac{\alpha_{k}K_{h}L}{N}\big|\tilde{\delta}_{ij}\big|\nonumber \\
 & \leq\frac{\alpha_{k}K_{h}L^{2}}{\sqrt{2}N^{2}}\label{eq:app-matrix-convergence}
\end{align}
where $d_{ij}=\|\bm{c}_{i,j}-\bm{s}_{k}\|$, $\tilde{\delta}_{ij}=\|\tilde{\bm{c}}_{i,j}-\bm{s}_{k}\|-d_{ij}$,
and $h'(d)=\lim_{t\downarrow0}\frac{1}{t}[h(d+t)-h(d)]$ denotes the
right derivative of $h(d)$. The upper bound (\ref{eq:app-matrix-convergence})
implies that $\bm{H}^{(k)}$ converges to $\tilde{\bm{H}}^{(k)}$,
in the sense that $\|\bm{H}^{(k)}-\tilde{\bm{H}}^{(k)}\|_{\text{F}}^{2}\leq N^{2}\big(\frac{\alpha_{k}K_{h}L^{2}}{\sqrt{2}N^{2}}\big)^{2}=\frac{C_{H}}{N^{2}}\to0$.
In addition, since $\bar{u}_{k}^{N}(y)$ and $\bar{v}_{k}^{N}(x)$
uniformly converge, we must have $\bar{u}_{k}^{N}(y)\to\tilde{u}_{k}^{N}(y)$
and $\bar{v}_{k}^{N}(x)\to\tilde{v}_{k}^{N}(x)$, uniformly as $N\to\infty$. 

Note that since $\tilde{\bm{u}}_{k}$ and $\tilde{\bm{v}}_{k}$ are
extracted from the common vector $\tilde{\bm{u}}_{k}^{\text{II}}$,
the interpolated functions $\tilde{u}_{k}^{N}(y)$ $\tilde{v}_{k}^{N}(x)$
has an identical shape. Therefore, there exists a unimodal symmetric
function $w(x)=w(-x)$, such that $\bar{u}_{k}^{N}(y)\to w(y-s_{k,2})$
and $\bar{v}_{k}^{N}(x)\to w(x-s_{k,1})$. The property $\int_{-\infty}^{\infty}w(x)^{2}dx=1$
is a direct consequence of the unit norm of singular vectors. \hfill\IEEEQED

\section{Proof of Theorem \ref{thm:squared-error-bound-full}}

\label{app:pf-thm-localization-error-bound}

We first compute the peak localization error bound given the signature
vector perturbations.

Let $\hat{\bm{\mathsf{v}}}_{1}$ be the dominant right singular vector
of $\bm{\mathsf{H}}$, the observation matrix in the case of conservative
construction, where $N=\sqrt{M}$. In the case of aggressive construction,
$N>\sqrt{M}$, let $\hat{\bm{\mathsf{v}}}_{1}$ be the dominant right
singular vector of $\hat{\bm{\mathsf{X}}}$, the solution to $\mathscr{P}2$.
Denote $\bm{\mathsf{e}}_{1}=\hat{\bm{\mathsf{v}}}_{1}-\bm{v}_{1}$.

From (\ref{eq:reflected-correlation}), it follows that 
\begin{align}
 & R(t;\hat{\bm{\mathsf{v}}}_{1})\nonumber \\
 & \quad=\int_{-\infty}^{\infty}\Big(w_{1}(x)+\bar{e}_{1}^{N}(x)\Big)\Big(w_{1}(-x+t)+\bar{e}_{1}^{N}(-x+t)\Big)dx\nonumber \\
 & \quad=\int_{-\infty}^{\infty}w_{1}(x)w_{1}(-x+t)dx+\int_{-\infty}^{\infty}w_{1}(x)\bar{e}_{1}^{N}(-x+t)dx\nonumber \\
 & \qquad+\int_{-\infty}^{\infty}\bar{e}_{1}^{N}(x)w_{1}(-x+t)dx+\int_{-\infty}^{\infty}\bar{e}_{1}^{N}(x)\bar{e}_{1}^{N}(-x+t)dx\nonumber \\
 & \quad=\tau(t-2s_{1,1})+\int_{-\infty}^{\infty}w_{1}(x)\bar{e}_{1}^{N}(-x+t)dx\nonumber \\
 & \qquad+\int_{+\infty}^{-\infty}\bar{e}_{1}^{N}(-y+t)w_{1}(y)(-dy)+E^{N}(t)\label{eq:app-approx-a}\\
 & \quad=\tau(t-2s_{1,1})+2\int_{-\infty}^{\infty}w_{1}(x)\bar{e}_{1}^{N}(-x+t)dx+E^{N}(t)\label{eq:app-pf-approx-R}
\end{align}
where Equation (\ref{eq:app-pf-approx-R}) uses the fact that $\int_{-\infty}^{\infty}w_{1}(x)w_{1}(-x+t)dx=\tau(t-2s_{1,1})$
from the integral (\ref{eq:reflected-correaltion-w}).

As $t=2\hat{\mathsf{s}}_{1,1}$ maximizes $R(t;\hat{\bm{\mathsf{v}}}_{1})$
in (\ref{eq:app-pf-approx-R}), we have 
\begin{align*}
 & \tau(2\hat{\mathsf{s}}_{1,1}-2s_{1,1})\\
 & +2\int_{-\infty}^{\infty}w_{1}(x)\bar{e}_{1}^{N}(-x+2\hat{\mathsf{s}}_{1,1})dx+E^{N}(2\hat{\mathsf{s}}_{1,1})-E^{N}(2s_{1,1})\\
 & \geq\tau(2s_{1,1}-2s_{1,1})+2\int_{-\infty}^{\infty}w_{1}(x)\bar{e}_{1}^{N}(-x+2s_{1,1})dx\\
 & =\tau(0)+2\int_{-\infty}^{\infty}w_{1}(x)\bar{e}_{1}^{N}(-x+2s_{1,1})dx
\end{align*}
where $\tau(0)=\int_{-\infty}^{\infty}w_{1}(x)^{2}dx=1$. As a result,
\begin{align*}
 & 1-\tau(2\hat{\mathsf{s}}_{1,1}-2s_{1,1})\\
 & \quad\leq2\int_{-\infty}^{\infty}w_{1}(x)\Big[e(-x+2\hat{\mathsf{s}}_{1,1})-e(-x+2s_{1,1})\Big]dx\\
 & \qquad+E^{N}(2s_{1,1})-E^{N}(2\hat{\mathsf{s}}_{1,1})\\
 & \quad\leq2C_{e}\big|(\bm{v}_{1}^{N})^{\text{T}}\bm{\mathsf{e}}_{1}^{N}\big|+o(\big|(\bm{v}_{1}^{N})^{\text{T}}\bm{\mathsf{e}}_{1}^{N}\big|)
\end{align*}
using conditions (\ref{eq:regularity-condition}) and (\ref{eq:regularity-condition-2})
for asymptotically large $N$. We obtain 
\begin{equation}
1-2C_{e}\big|(\bm{v}_{1}^{N})^{\text{T}}\bm{\mathsf{e}}_{1}^{N}\big|+o(\big|(\bm{v}_{1}^{N})^{\text{T}}\bm{\mathsf{e}}_{1}^{N}\big|)\leq\tau(2|\hat{\mathsf{s}}_{1,1}-s_{1,1}|)\label{eq:app-pf-error-bound-a0}
\end{equation}
which leads to $\big|\hat{\mathsf{s}}_{1,1}-s_{1,1}\big|\leq\frac{1}{2}\tau^{-1}\big(1-2C_{e}\big|(\bm{v}_{1}^{N})^{\text{T}}\bm{\mathsf{e}}_{1}^{N}\big|+o(\big|(\bm{v}_{1}^{N})^{\text{T}}\bm{\mathsf{e}}_{1}^{N}\big|\big)$,
where $t=\tau^{-1}(y)$ is the inverse function of $y=\tau(t)$, $t\geq0$.
Hence, 
\begin{equation}
\big(\hat{\mathsf{s}}_{1,1}-s_{1,1}\big)^{2}\leq\frac{1}{4}\Big[\tau^{-1}\big(1-2C_{e}\phi_{0}+o(\phi_{0})\big)\Big]^{2}\label{eq:app-pf-error-bound-a}
\end{equation}
where we denote $\phi_{0}\triangleq\big|(\bm{v}_{1}^{N})^{\text{T}}\bm{\mathsf{e}}_{1}^{N}\big|$. 

In addition, from the assumption that the first and the second-order
derivatives, $\tau'(t)$ and $\tau''(t)$, exist and are continuous
at $t=0$, we bound $\tau(t)$ by a quadratic function 
\begin{equation}
\tau(t)\leq\tau(0)+\tau'(0)t+\frac{1}{2}(\tau''(0)+\nu)t^{2}\label{eq:app-upper-bound-tau-t}
\end{equation}
for some small $\nu>0$. Note that for any $\nu$, there exists $a_{\nu}>0$,
such that the quadratic upper bound (\ref{eq:app-upper-bound-tau-t})
holds for all $t\in[0,a_{\nu}]$. In addition, from the definition
of $\tau(t)$ in (\ref{eq:autocorrelation-function}), we obtain $\tau(0)'=0$
and $\tau''(0)<0$. Then, from (\ref{eq:app-pf-error-bound-a0}),
we arrive at 
\begin{align}
 & 1-2C_{e}\phi_{0}+o(\phi_{0})\leq\tau(2|\hat{\mathsf{s}}_{1,1}-s_{1,1}|)\nonumber \\
 & \qquad\qquad\qquad\leq1+\frac{1}{2}(\tau''(0)+\nu)\times4\big(\hat{\mathsf{s}}_{1,1}-s_{1,1}\big)^{2}\label{eq:app-pf-error-bound-b}
\end{align}
if $t=2\big|\hat{\mathsf{s}}_{1,1}-s_{1,1}\big|$ is small enough,
\emph{i.e.}, $t\in[0,a_{\nu}]$. Note that $\tau''(0)+\nu<0$ as $\nu$
can be chosen arbitrarily small. As a result, 
\[
\big|\hat{\mathsf{s}}_{1,1}-s_{1,1}\big|^{2}\leq\frac{1}{2}\frac{\big(2C_{e}\phi_{0}+o(\phi_{0})\big)}{-\tau^{''}(0)-\nu}=\frac{1}{2}\frac{\big(2C_{e}+o(1)\big)\phi_{0}}{-\tau^{''}(0)-\nu}.
\]
For small enough $\phi_{0}$, we will have $|o(\phi_{0})|<C_{e}$,
which yields 
\begin{equation}
\big|\hat{\mathsf{s}}_{1,1}-s_{1,1}\big|^{2}\leq\frac{1}{2}\cdot\frac{3C_{e}\phi_{0}}{-\tau^{''}(0)-\nu}.\label{eq:app-pf-error-bound-c}
\end{equation}

A similar result can be derived for the error component $(\hat{\mathsf{s}}_{1,2}-s_{1,2})^{2}$
by analyzing the perturbation of $\bm{u}_{1}$, and the result is
statistically identical to (\ref{eq:app-pf-error-bound-a}) and (\ref{eq:app-pf-error-bound-c}).
Combining (\ref{eq:app-pf-error-bound-a}) and (\ref{eq:app-pf-error-bound-c})
yields 
\begin{align}
\|\hat{\bm{\mathsf{s}}}_{1}-\bm{s}_{1}\|_{2}^{2} & \leq\max\Big\{\frac{3C_{e}\phi_{0}}{-\tau^{''}(0)-\nu},\nonumber \\
 & \qquad\qquad\frac{1}{2}\Big[\tau^{-1}(1-2C_{e}\phi_{0}+o(\phi_{0}))\big]^{2}\Big\}.\label{eq:app-mse-long}
\end{align}
Moreover, for small enough $\phi_{0}=\big|(\bm{v}_{1}^{N})^{\text{T}}\bm{\mathsf{e}}_{1}^{N}\big|$,
(\ref{eq:app-mse-long}) simplifies to
\begin{equation}
\|\hat{\bm{\mathsf{s}}}_{1}-\bm{s}_{1}\|_{2}^{2}\leq\frac{3C_{e}\phi_{0}}{-\tau^{''}(0)-\nu}.\label{eq:app-mse}
\end{equation}

We hereafter drop the superscript $N$ from $\bm{v}_{1}^{N}$ and
$\bm{\mathsf{e}}_{1}^{N}$ for simplicity and derive the signature
perturbation $\phi_{0}=\big|\bm{v}_{1}^{\text{T}}\bm{\mathsf{e}}_{1}\big|$
by examining two cases. 

\subsection{The Case of Conservative Construction $N=\sqrt{M}$}

Consider that the sensing location $\bm{z}^{(m)}$ is inside the grid
cell centered at $\bm{c}_{i,j}$. Then, according to the matrix observation
model (\ref{eq:measurement-model})\textendash (\ref{eq:matrix-observation-model-true}),
we have 
\begin{align}
\mathsf{H}_{ij}-H_{ij} & =\frac{L}{N}\Big(\mathsf{h}^{(m)}-\alpha_{1}h(d_{ij})\Big)\nonumber \\
 & =\frac{\alpha_{1}L}{N}\Big(\int_{0}^{1}h'(d_{ij}+t\delta^{(m)})dt+\frac{\mathsf{n}^{(m)}}{\alpha_{1}}\Big)\nonumber \\
 & =\frac{\alpha_{1}L}{N}\Big(\Xi_{ij}+\frac{\mathsf{n}^{(m)}}{\alpha_{1}}\Big)\label{eq:app-pf-per-entry-error-0}
\end{align}
where $d_{ij}=\|\bm{c}_{i,j}-\bm{s}_{1}\|_{2}$, $\delta^{(m)}=\|\bm{z}^{(m)}-\bm{s}_{1}\|_{2}-d_{ij}$,
$h'(d)=\lim_{t\downarrow0}\frac{1}{t}[h(t+d)-h(d)]$ is the right
derivative of $h(d)$, and $\Xi_{ij}\triangleq\int_{0}^{1}h'(d_{ij}+t\delta^{(m)})dt$. 

Note that by the Lipschitz continuity of $h(d)$, we have 
\begin{equation}
\big|\Xi_{ij}\big|\leq K_{h}\delta^{(m)}\leq K_{h}L/(\sqrt{2}N)\label{eq:app-grid-sensitivity}
\end{equation}
and therefore, $\Xi_{ij}$ is a \emph{sub-Gaussian} random variable
with sub-Gaussian moment bounded by $K_{h}L/(\sqrt{2}N)$ \cite{buldygin1980sub,stromberg1994probability}.\footnote{Recall that a real-valued random variable $X$ is said to be sub-Gaussian
if it has the property that there is some $b>0$ such that for every
$t\in\mathbb{R}$ one has $\mathbb{E}\{e^{tX}\}\leq\exp(\frac{1}{2}b^{2}t^{2})$.
The sub-Gaussian moment of $X$ is defined as $s_{\text{G}}(X)=\inf\{b\geq0\big|\mathbb{E}\{e^{tX}\}\leq\exp(\frac{1}{2}b^{2}t^{2}),\forall t\in\mathbb{R}\}$,
which has property $s_{\text{G}}(aX)=|a|s_{\text{G}}(X)$ \cite{buldygin1980sub,stromberg1994probability}. } 

In addition, the bounded random variable $\mathsf{n}^{(m)}$ is also
sub-Gaussian with sub-Gaussian moment bounded by $\bar{\sigma}_{\text{n}}$.
As a result, the random quantity $\frac{\alpha_{1}L}{N}\Xi_{ij}+\frac{\alpha_{1}L}{N}\frac{\mathsf{n}^{(m)}}{\alpha_{1}}$
from (\ref{eq:app-pf-per-entry-error-0}) is sub-Gaussian with sub-Gaussian
moment bounded by\footnote{If $X$ and $Y$ are independent sub-Gaussian, the sub-Gaussian moment
of $X+Y$ satisfies $s_{\text{G}}(X+Y)=\sqrt{s_{\text{G}}(X)+s_{\text{G}}(Y)}$
\cite{buldygin1980sub,stromberg1994probability}. }
\begin{equation}
\bar{\omega}\triangleq\frac{\alpha_{1}L}{N}\sqrt{(\frac{K_{h}L}{\sqrt{2}N})^{2}+(\frac{\bar{\sigma}_{\text{n}}}{\alpha_{1}})^{2}}.\label{eq:app-pf-per-entry-error-2nd-moment}
\end{equation}

Consider the $N\times N$ matrix of the sample error $\bm{\mathsf{E}}\triangleq\bm{\mathsf{H}}-\bm{H}$,
where the entries $\mathsf{E}_{ij}$ have zero mean and sub-Gaussian
moment bounded by $\bar{\omega}$ in (\ref{eq:app-pf-per-entry-error-2nd-moment}).
This implies that $\tilde{\bm{\mathsf{E}}}=\frac{1}{\bar{\omega}}\bm{\mathsf{E}}$
have zero mean entries with sub-Gaussian moment bounded by $1$. A
bound of the spectral norm $\sigma(\tilde{\bm{\mathsf{E}}})$ can
be derived using the following result. 
\begin{lem}
[Spectral Norm \cite{rudelson2010non}]\label{lem:Spectral-Norm}
For an $N\times n$ random matrix $\bm{X}$ whose entries are independent
zero mean sub-Gaussian random variables with sub-Gaussian moments
bounded by 1, it holds that
\[
\mathbb{P}\big\{\sigma(\bm{X})>C(\sqrt{N}+\sqrt{n})+t\big\}\leq2e^{-ct^{2}},\quad t\geq0
\]
for some universal constants $C$ and $c$.
\end{lem}

By choosing $N=n$ and $t=C\sqrt{N}$ in Lemma \ref{lem:Spectral-Norm},
we obtain $\sigma(\tilde{\bm{\mathsf{E}}})\leq3C\sqrt{N}$ with probability
at least $1-2e^{-cC^{2}N}$. Therefore, 
\begin{equation}
\sigma(\bm{\mathsf{E}})^{2}\leq9C^{2}\bar{\omega}^{2}N=C_{0}\frac{\alpha_{1}^{2}L^{2}}{N}\Big(\frac{K_{h}^{2}L^{2}}{2N^{2}}+\frac{\bar{\sigma}_{\text{n}}^{2}}{\alpha_{1}^{2}}\Big)\label{eq:spectral-norm-full}
\end{equation}
holds with probability at least $1-2e^{-C_{3}N}$, where $C_{0}=9C^{2}$,
and $C_{3}=cC^{2}$. 

% -----------

We now derive an upper bound of $\big|\bm{v}_{1}^{\text{T}}\bm{\mathsf{e}}_{1}\big|$.

Note that $\bm{v}_{1}^{\text{T}}\bm{\mathsf{e}}_{1}<0$ because $1=\|\hat{\bm{\mathsf{v}}}_{1}\|_{2}^{2}=\|\bm{v}_{1}+\bm{\mathsf{e}}_{1}\|_{2}^{2}=1+2\bm{v}_{1}^{\text{T}}\bm{\mathsf{e}}_{1}+\|\bm{\mathsf{e}}_{1}^{\text{T}}\bm{\mathsf{e}}_{1}\|_{2}^{2}$,
and hence, $2\bm{v}_{1}^{\text{T}}\bm{\mathsf{e}}_{1}=-\|\bm{\mathsf{e}}_{1}^{\text{T}}\bm{\mathsf{e}}_{1}\|_{2}^{2}<0$.
Then, the singular vector perturbation can be obtained as 
\begin{align}
\sin\angle(\bm{v}_{1},\hat{\bm{\mathsf{v}}}_{1}) & =\sqrt{1-\big|\bm{v}_{1}^{\text{T}}(\bm{v}_{1}+\bm{\mathsf{e}}_{1})\big|^{2}}\nonumber \\
 & =\sqrt{-2\bm{v}_{1}^{\text{T}}\bm{\mathsf{e}}_{1}+(\bm{v}_{1}^{\text{T}}\bm{\mathsf{e}}_{1})^{2}}\nonumber \\
 & \geq\sqrt{2\big|\bm{v}_{1}^{\text{T}}\bm{\mathsf{e}}_{1}\big|}\label{eq:app-pf-sin}
\end{align}
where $|\cdot|$ denotes the absolute value operator. 

On the other hand, using the singular vector perturbation results
in \cite{vu:J11singular}, we know that 
\begin{equation}
\sin\angle(\bm{v}_{1},\hat{\bm{\mathsf{v}}}_{1})\leq\frac{2\sigma(\bm{\mathsf{E}})}{\sigma_{1,1}-\sigma_{1,2}}\label{eq:app-pf-sin-upperbound}
\end{equation}
where $\sigma_{1}$ and $\sigma_{2}$ are the first and second dominant
singular values of $\bm{H}$. Recall that we have $\sigma_{1,1}-\sigma_{1,2}\to\kappa\alpha_{1}$,
as $N\to\infty$.

Therefore, we arrive at 
\begin{align*}
\big|\bm{v}_{1}^{\text{T}}\bm{\mathsf{e}}_{1}\big| & \leq\frac{1}{2}\big(\sin\angle(\bm{v}_{1},\hat{\bm{\mathsf{v}}}_{1})\big)^{2}\\
 & \leq\frac{1}{2}\Big(\frac{2\sigma(\bm{\mathsf{E}})}{\kappa\alpha_{1}}\Big)^{2}\\
 & \leq\frac{C_{0}K_{h}^{2}L^{4}}{\kappa^{2}N^{3}}+\frac{2C_{0}L^{2}}{\kappa^{2}N}\frac{\bar{\sigma}_{\text{n}}^{2}}{\alpha_{1}^{2}}
\end{align*}
for asymptotically large $N$, with probability at least $1-2e^{-C_{3}N}$.
As $M=N^{2}$, it further holds that 
\begin{equation}
3C_{e}\big|\bm{v}_{1}^{\text{T}}\bm{\mathsf{e}}_{1}\big|\leq\frac{3C_{0}C_{e}K_{h}^{2}L^{4}}{\kappa^{2}M^{1.5}}+\frac{6C_{0}C_{e}L^{2}}{\kappa^{2}\sqrt{M}}\frac{\bar{\sigma}_{\text{n}}^{2}}{\alpha_{1}^{2}}\triangleq\phi_{1}\label{eq:app-pf-singular-vector-perturbation}
\end{equation}
for large $M$ with probability at least $1-2e^{-C_{3}N}$. 

Note that for large enough $M$, the term $\phi_{1}$ in (\ref{eq:app-pf-singular-vector-perturbation})
becomes small enough to satisfy $t=\tau^{-1}(1-\phi_{1})\leq a_{\nu}$,
such that $\tau(t)\leq\tau(0)+\tau'(0)t+\frac{1}{2}(\tau''(0)+\nu)t^{2}$
holds, and therefore, (\ref{eq:app-mse-long}) simplifies to (\ref{eq:app-mse}).
Applying $3C_{e}\phi_{0}\leq\phi_{1}$ to (\ref{eq:app-mse}) and
using (\ref{eq:app-pf-singular-vector-perturbation}), we obtain (\ref{eq:squared-error-bound-full}). 

% -----------

\subsection{The Case of Aggressive Construction $N>\sqrt{M}$}

We first compute the total measurement noise using property (\ref{eq:app-grid-sensitivity})
and the bounded noise assumption in Theorem \ref{thm:squared-error-bound-partial}.
From (\ref{eq:app-pf-per-entry-error-0}), we have 
\begin{align}
\big(\mathsf{H}_{ij}-H_{ij}\big)^{2} & =\frac{\alpha_{1}^{2}L^{2}}{N^{2}}\Big(\Xi^{2}+2\Xi\frac{\mathsf{n}^{(m)}}{\alpha_{1}}+\big(\frac{\mathsf{n}^{(m)}}{\alpha_{1}}\big)^{2}\Big)\nonumber \\
 & \leq\frac{\alpha_{1}^{2}L^{2}}{N^{2}}\Big(\frac{K_{h}^{2}L^{2}}{2N^{2}}+\frac{\sqrt{2}K_{h}L}{N}\frac{\bar{\sigma}_{\text{n}}}{\alpha_{1}}+\frac{\bar{\sigma}_{\text{n}}^{2}}{\alpha_{1}^{2}}\Big)\triangleq\bar{\varepsilon}^{2}.\label{eq:app-bar-epsilon}
\end{align}
As a result, 
\begin{align*}
\|\mathcal{P}_{\Omega}(\bm{\mathsf{H}}-\bm{H})\|_{\text{F}}^{2} & =\sum_{(i,j)\in\Omega}|\mathsf{H}_{ij}-H_{ij}|^{2}\leq M\bar{\varepsilon}^{2}.
\end{align*}

We then characterize the completed matrix $\hat{\bm{\mathsf{X}}}$
as the solution to $\mathscr{P}2$. 
\begin{lem}
[Matrix completion with noise \cite{CanPla:J10}]\label{prop:matrix-completion-noise}
Suppose that the parameter $\epsilon$ in $\mathscr{P}2$ is chosen
to satisfy $\epsilon\geq\|\mathcal{P}_{\Omega}(\bm{\mathsf{H}}-\bm{H})\|_{\text{F}}$.
In addition, assume that $M\geq CN(\log N)^{2}$ for some constant
$C=C'\beta$. Then, with probability at least $1-N^{-\beta}$, 
\begin{equation}
\|\hat{\bm{\mathsf{X}}}-\bm{H}\|_{\text{F}}\leq4\sqrt{\frac{(2+p)N}{p}}\epsilon+2\epsilon\label{eq:matrix-completion-noise-bound}
\end{equation}
where $p=M/N^{2}$.
\end{lem}

According to Lemma \ref{prop:matrix-completion-noise}, we choose
the parameter $\epsilon=\sqrt{M}\bar{\varepsilon}$ in $\mathscr{P}2$
for (\ref{eq:matrix-completion-noise-bound}) to hold. In addition,
the bound in (\ref{eq:matrix-completion-noise-bound}) can be simplified
as 
\begin{align}
\|\hat{\bm{\mathsf{X}}}-\bm{H}\|_{\text{F}} & \leq4\sqrt{\frac{(2+p)N}{p}}\epsilon+2\epsilon\nonumber \\
 & =\bigg(4\sqrt{\frac{2N^{3}}{M}+N}+2\bigg)\epsilon\label{eq:bar-eta-eq1}\\
 & \leq\bigg(4\sqrt{\frac{2N^{3}}{CN(\log N)^{2}}+N}+2\bigg)\epsilon\label{eq:bar-eta-eq2}\\
 & =\bigg(\sqrt{\frac{32}{C}}\frac{N}{\log N}+o(\sqrt{N})\bigg)\sqrt{M}\bar{\varepsilon}\label{eq:bar-eta-eq3}\\
 & \qquad\qquad\qquad\qquad\qquad\qquad\qquad\triangleq\bar{\eta}\nonumber 
\end{align}
where equality (\ref{eq:bar-eta-eq1}) is based on the relation $p=M/N^{2}$,
inequality (\ref{eq:bar-eta-eq2}) is due to $M\geq CN(\log N)^{2}$,
and equality (\ref{eq:bar-eta-eq3}) can be shown using the following
lemma. 
\begin{lem}
\label{lem:app-square-root-approximation}Suppose that non-negative
functions $f$ and $g$ satisfy $f(x)\to\infty$, $g(x)\to\infty$,
and $h(x)=f(x)/g(x)\to\infty$, as $x\to\infty$. Then, $\sqrt{f(x)+g(x)}=\sqrt{f(x)}+o(\sqrt{g(x)})$,
where $\lim_{x\to\infty}o(\sqrt{g(x)})/\sqrt{g(x)}=0$. 
\end{lem}
\begin{proof}
We can show that the residual error $\sqrt{f(x)+g(x)}-\sqrt{f(x)}$
grows slower than $\sqrt{g(x)}$ as follows
\begin{align*}
 & \lim_{x\to\infty}\frac{\sqrt{f(x)+g(x)}-\sqrt{f(x)}}{\sqrt{g(x)}}\\
 & \quad=\lim_{x\to\infty}\sqrt{\frac{f(x)}{g(x)}+1}-\sqrt{\frac{f(x)}{g(x)}}\\
 & \quad=\lim_{x\to\infty}\frac{(\sqrt{h(x)+1}-\sqrt{h(x)})(\sqrt{h(x)+1}+\sqrt{h(x)})}{\sqrt{h(x)+1}+\sqrt{h(x)}}\\
 & \quad=\lim_{x\to\infty}\frac{1}{\sqrt{h(x)+1}+\sqrt{h(x)}}\\
 & \quad=0.
\end{align*}
\end{proof}

Let 
\[
f(N)=\frac{2N^{3}}{CN(\log N)^{2}}=\frac{2N^{2}}{C(\log N)^{2}}
\]
 and $g(N)=N$. We have $f,g\to\infty$ and $f/g\to\infty$ for $N\to\infty$.
Then, equality (\ref{eq:bar-eta-eq3}) follows from Lemma \ref{lem:app-square-root-approximation}. 

Consider $\hat{\bm{\mathsf{E}}}$ as a random matrix whose entries
are independent zero mean with sub-Gaussian moments bounded by $\tilde{\omega}\triangleq\bar{\eta}/N$.
Since $\frac{1}{N^{2}}\sum_{i,j}\big|\hat{X}_{ij}-H_{ij}\big|^{2}\leq\bar{\eta}^{2}/N^{2}$
with probability at least $1-N^{-\beta}$, we know $\sigma(\hat{\bm{\mathsf{X}}}-\bm{H})\leq\sigma(\hat{\bm{\mathsf{E}}})$
with probability at least $1-N^{-\beta}$. Following the same derivation
as for (\ref{eq:spectral-norm-full}), we obtain 
\begin{align*}
\sigma(\hat{\bm{\mathsf{E}}})^{2}\leq C_{0}\tilde{\omega}^{2}N & =\frac{C_{0}}{N}\Big(\frac{32N^{2}}{C(\log N)^{2}}+o(N)\Big)CN(\log N)^{2}\bar{\varepsilon}^{2}\\
 & =\bar{\varepsilon}^{2}\Big(32C_{0}N^{2}+o\big(N(\log N)^{2}\big)\Big)
\end{align*}
with probability at least $(1-2e^{-C_{3}N})(1-N^{-\beta})$, \emph{i.e.},
$1-\mathcal{O}(N^{-\beta})$.

Using the singular vector under perturbation results in \cite{vu:J11singular}
and the expression for $\bar{\epsilon}$ in (\ref{eq:app-bar-epsilon}),
and following similar calculations in (\ref{eq:app-pf-sin})\textendash (\ref{eq:app-pf-singular-vector-perturbation}),
we obtain 
\begin{align}
\big|\bm{v}_{1}^{\text{T}}\bm{\mathsf{e}}_{1}\big| & \leq\frac{1}{2}\Big(\frac{2\sigma(\hat{\bm{\mathsf{E}}})}{\kappa\alpha_{1}}\Big)^{2}\nonumber \\
 & =\frac{2L^{2}}{\kappa^{2}}\Big(\frac{K_{h}^{2}L^{2}}{2N^{2}}+\frac{\sqrt{2}K_{h}L}{N}\frac{\bar{\sigma}_{\text{n}}}{\alpha_{1}}+\frac{\bar{\sigma}_{\text{n}}^{2}}{\alpha_{1}^{2}}\Big)\label{eq:app-pf-singular-perturbation-partial}\\
 & \qquad\qquad\qquad\qquad\times\Big(32C_{0}+o\big(\frac{(\log N)^{2}}{N}\big)\Big)\nonumber 
\end{align}
with probability $1-\mathcal{O}(N^{-\beta})$.

Finally, if one choose parameters $M,N$ such that $M$ is the smallest
integer satisfying $M\geq CN(\log N)^{2}$, then we can write $M=CN(\log N)^{2}+\epsilon_{0}$,
where $\epsilon_{0}\in[0,1)$. As a result, from $M^{1-\alpha}=[CN(\log N)^{2}+\epsilon_{0}]^{1-\alpha}$,
we arrive at 
\begin{equation}
\frac{M^{1-\alpha}}{N}=C^{1-\alpha}\frac{(\log N)^{2(1-\alpha)}}{N^{\alpha}}\bigg[1+\frac{\epsilon_{0}}{CN(\log N)^{2}}\bigg]^{1-\alpha}.\label{eq:app-M-N-relation}
\end{equation}
It is clear that for $0<\alpha<\frac{1}{2}$, the term $\frac{(\log N)^{2(1-\alpha)}}{N^{\alpha}}$
converges to $0$ as $N\to\infty$. As a result, there exists a finite
integer $N_{1}<\infty$ such that for any $N>N_{1}$, the right hand
side of (\ref{eq:app-M-N-relation}) is less than $1$, \emph{i.e.},
$\frac{M^{1-\alpha}}{N}\leq1$. This implies that $\frac{1}{N}\leq\frac{1}{M^{1-\alpha}}$.
Similarly, using $M^{2-\alpha}=[CN(\log N)^{2}+\epsilon_{0}]^{2-\alpha}$,
one can obtain $\frac{1}{N^{2}}\leq\frac{1}{M^{2-\alpha}}$ for asymptotically
large $M,N$. Substituting $\frac{1}{N}\leq\frac{1}{M^{1-\alpha}}$
and $\frac{1}{N^{2}}\leq\frac{1}{M^{2-\alpha}}$ to (\ref{eq:app-pf-singular-perturbation-partial})
and omitting the $o(\cdot)$ term, one can obtain 
\begin{equation}
3C_{e}\big|\bm{v}_{1}^{\text{T}}\bm{\mathsf{e}}_{1}\big|\leq\frac{C_{0}'C_{e}L^{2}}{\kappa^{2}}\Big(\frac{K_{h}^{2}L^{2}}{2M^{2-\alpha}}+\frac{\sqrt{2}K_{h}L}{M^{1-\alpha}}\frac{\bar{\sigma}_{\text{n}}}{\alpha_{1}}+\frac{\bar{\sigma}_{\text{n}}^{2}}{\alpha_{1}^{2}}\Big)\triangleq\phi_{2}\label{eq:app-pf-singular-perturbation-partial-phi2}
\end{equation}
with probability $1-\mathcal{O}(N^{-\beta})$, where $C_{0}'=192C_{0}$. 

Note that for large enough $M$ and high SNR $\alpha_{1}/\bar{\sigma}_{\text{n}}$,
the term $\phi_{2}$ in (\ref{eq:app-pf-singular-perturbation-partial-phi2})
becomes small enough to satisfy $t=\tau^{-1}(1-\phi_{2})\leq a_{\nu}$,
such that $\tau(t)\leq\tau(0)+\tau'(0)t+\frac{1}{2}(\tau''(0)+\nu)t^{2}$
holds, and therefore, (\ref{eq:app-mse-long}) simplifies to (\ref{eq:app-mse}).
Applying $3C_{e}\phi_{0}\leq\phi_{2}$ to (\ref{eq:app-mse}) and
using (\ref{eq:app-pf-singular-perturbation-partial-phi2}), we obtain
(\ref{eq:squared-error-bound-partial}).

% ----- DONE -----

\section{Proof of Theorem \ref{thm:Unique-local-maximum}}

\label{app:pf-thm-unimodal-local-maximum}

\Ac{wlog}, assume that the two sources are located at $\bm{s}_{1}=(0,0)$
and $\bm{s}_{2}=(D\cos\theta,D\sin\theta)$. Define $w_{\text{c}}(x,\theta)=w(x-D\cos\theta)$
and $w_{\text{s}}(x,\theta)=w(x-D\sin\theta)$, where $w(x)$ is the
unimodal and symmetric function defined in Proposition \ref{prop:symmetry-signature-vector}. 

Using Proposition \ref{prop:symmetry-signature-vector}, we have the
following approximation under large $N$ 
\begin{align*}
\|\bm{u}_{1}+\bm{u}_{2}\|_{2}^{2} & \approx\int_{-\infty}^{\infty}\big(w(x)+w_{\text{c}}(x,\theta)\big)^{2}dx\\
 & \triangleq\langle(w+w_{\text{c}})^{2}\rangle
\end{align*}
where we have defined an integration operator $\left\langle \cdot\right\rangle $
as 
\[
\left\langle f\right\rangle \triangleq\int_{-\infty}^{\infty}f(x,\theta)dx
\]
for a function $f(x,\theta)$. Note that the operator $\left\langle \cdot\right\rangle $
is linear and satisfies the additive property, \emph{i.e.}, $\langle af\rangle=a\langle f\rangle$
and $\langle f+g\rangle=\langle f\rangle+\langle g\rangle$, for a
constant $a$ and a function $g(x,\theta)$.

Similarly, $\|\bm{v}_{1}+\bm{v}_{2}\|_{2}^{2}\approx\langle(w+w_{\text{s}})^{2}\rangle$,
$\|\bm{u}_{1}-\bm{u}_{2}\|_{2}^{2}\approx\langle(w-w_{\text{c}})^{2}\rangle$,
and $\|\bm{v}_{1}-\bm{v}_{2}\|_{2}^{2}\approx\langle(w-w_{\text{s}})^{2}\rangle$.

In the case of $\alpha_{1}=\alpha_{2}$, it can be shown that the
\ac{svd} of $\bm{H}$ is given by 
\begin{equation}
\bm{H}=\sigma_{1}\bm{p}_{1}\bm{q}_{1}^{\text{T}}+\sigma_{2}\bm{p}_{2}\bm{q}_{2}^{\text{T}}\label{eq:SVD-H-two-source}
\end{equation}
where $\sigma_{1}(\theta)=\frac{1}{2}\alpha_{1}\|\bm{u}_{1}+\bm{u}_{2}\|_{2}\|\bm{v}_{1}+\bm{v}_{2}\|_{2}$
and $\sigma_{2}(\theta)=\frac{1}{2}\alpha_{1}\|\bm{u}_{1}-\bm{u}_{2}\|_{2}\|\bm{v}_{1}-\bm{v}_{2}\|_{2}$
are the singular values, and 
\[
\bm{p}_{1}=\frac{\bm{u}_{1}+\bm{u}_{2}}{\|\bm{u}_{1}+\bm{u}_{2}\|_{2}},\quad\bm{q}_{1}=\frac{\bm{v}_{1}+\bm{v}_{2}}{\|\bm{v}_{1}+\bm{v}_{2}\|_{2}}
\]
\[
\bm{p}_{2}=\frac{\bm{u}_{1}-\bm{u}_{2}}{\|\bm{u}_{1}-\bm{u}_{2}\|_{2}},\quad\bm{q}_{2}=\frac{\bm{v}_{1}-\bm{v}_{2}}{\|\bm{v}_{1}-\bm{v}_{2}\|_{2}}
\]
are the corresponding singular vectors. Here, all the components are
functions of $\theta$. 

As a result, 
\begin{align}
\mu(\theta) & \triangleq\frac{\sigma_{2}(\theta)^{2}}{\sigma_{1}(\theta)^{2}}\nonumber \\
 & \approx\frac{\langle(w-w_{\text{c}})^{2}\rangle\langle(w-w_{\text{s}})^{2}\rangle}{\langle(w+w_{\text{c}})^{2}\rangle\langle(w+w_{\text{s}})^{2}\rangle}\nonumber \\
 & =\frac{\big(1-\langle w\cdot w_{\text{c}}\rangle\big)\big(1-\langle w\cdot w_{\text{s}}\rangle\big)}{\big(1+\langle w\cdot w_{\text{c}}\rangle\big)\big(1+\langle w\cdot w_{\text{s}}\rangle\big)}\label{eq:mu-function}
\end{align}
where we have used the fact that $\langle(w-w_{\text{c}})^{2}\rangle=\langle w^{2}\rangle+\langle w_{\text{c}}^{2}\rangle-2\langle w\cdot w_{\text{c}}\rangle=2\big(1-\langle w\cdot w_{\text{c}}\rangle\big)$.

In addition, from properties of calculus, if $f(x,\theta)$ and $\frac{\partial}{\partial\theta}f(x,\theta)$
are continuous in $\theta,$ then 
\begin{align*}
\frac{d}{d\theta}\left\langle f\right\rangle  & =\frac{d}{d\theta}\int_{-\infty}^{\infty}f(x,\theta)dx\\
 & =\int_{-\infty}^{\infty}\frac{\partial}{\partial\theta}f(x,\theta)dx=\Big\langle\frac{\partial}{\partial\theta}f\Big\rangle.
\end{align*}
Therefore, defining 
\begin{align*}
w_{\text{c}}'(x,\theta) & \triangleq\frac{d}{dx}w(x)\big|_{x=x-D\cos\theta}\\
w_{\text{s}}'(x,\theta) & \triangleq\frac{d}{dx}w(x)\big|_{x=x-D\sin\theta}
\end{align*}
we have
\begin{align*}
\frac{d}{d\theta}\langle w\cdot w_{\text{c}}\rangle & =\langle w\cdot\frac{\partial}{\partial\theta}w_{\text{c}}(x,\theta)\rangle=\langle w\cdot w_{\text{c}}'\rangle D\sin\theta\\
\frac{d}{d\theta}\langle w\cdot w_{\text{s}}\rangle & =\langle w\cdot\frac{\partial}{\partial\theta}w_{\text{s}}(x,\theta)\rangle=-\langle w\cdot w_{\text{s}}'\rangle D\cos\theta.
\end{align*}

With some algebra, the derivative of $\mu(\theta)$ can be obtained
as 
\begin{align*}
\frac{d}{d\theta}\mu(\theta) & =\eta\Big[D\cos\theta\langle w\cdot w_{\text{s}}'\rangle\big(1-\langle w\cdot w_{\text{c}}\rangle^{2}\big)\\
 & \qquad\qquad-D\sin\theta\langle w\cdot w_{\text{c}}'\rangle\big(1-\langle w\cdot w_{\text{s}}\rangle^{2}\big)\Big]\\
 & =\eta\Big[-t\cdot\tau'(s)\big(1-\tau(t)^{2}\big)+s\cdot\tau'(t)\big(1-\tau(s)^{2}\big)\Big]
\end{align*}
where $\eta=2\big(1+\langle w\cdot w_{\text{c}}\rangle\big)^{-2}\big(1+\langle w\cdot w_{\text{s}}\rangle\big)^{-2}$,
$t=D\cos\theta$, and $s=D\sin\theta$. 

Note that $0<s<t$ for $0<\theta<\frac{\pi}{4}$. Applying condition
(\ref{eq:correlation-condition}), we have 
\begin{align*}
\frac{d}{d\theta}\mu(\theta) & >\eta\cdot t\cdot\tau'(s)\Big[\big(1-\tau(s)^{2}\big)-\big(1-\tau(t)^{2}\big)\Big]\\
 & =\eta\cdot t\cdot\tau'(s)\big(\tau(t)^{2}-\tau(s)^{2}\big)\\
 & >0
\end{align*}
since $\tau'(s)<0$ and $\tau(t)<\tau(s)$ for $0<s<t$.

This confirms that $\mu(\theta)$ is a strictly increasing function,
and hence $\rho(\theta)$ is a strictly decreasing function in $\theta\in(0,\frac{\pi}{4})$.
The result is thus proved.\hfill\IEEEQED

\section{Proof of Proposition \ref{prop:Partial-convergence}}

\label{app:pf-partial-convergence}

For notation brevity, we drop the symbol $t$ for the variables related
to the continuous-time algorithm dynamic $\bm{X}(t)$ wherever the
meaning is clear.

Denote the Hessian function of $f(\bm{X})$ along the direction $\bm{\xi}\in\mathbb{R}^{N\times N}$
as 
\[
h(\bm{\xi},\bm{X})=\lim_{\gamma\to0}\frac{1}{\gamma}\Big[g(\bm{X}+\gamma\bm{\xi})-g(\bm{X})\Big].
\]
Then, a Taylor's expansion of the gradient function $g(\bm{X})$ yields
\[
g(\bm{X})=g(\hat{\bm{\mathsf{X}}})+h(s\bm{\xi},\hat{\bm{\mathsf{X}}})+o(s)
\]
where $\bm{\xi}=\frac{1}{\gamma}(\bm{X}-\hat{\bm{\mathsf{X}}})$ and
$\gamma=\|\bm{X}-\hat{\bm{\mathsf{X}}}\|_{\text{F}}$. Therefore,
as $g(\hat{\bm{\mathsf{X}}})=\mathbf{0}$, it holds that $g(\bm{X})\approx h(\bm{\mathsf{X}}_{e},\hat{\bm{\mathsf{X}}})$
for small $s$. 

In addition, it holds that 
\begin{align*}
\frac{d}{dt}\mathcal{E}(\bm{\mathsf{X}}_{e}(t)) & =\text{tr}\Big\{\big(\bm{X}(t)-\hat{\bm{\mathsf{X}}}\big)^{\text{T}}\frac{d}{dt}\bm{X}(t)\Big\}\\
 & =-\text{tr}\Big\{\big(\bm{X}(t)-\hat{\bm{\mathsf{X}}}\big)^{\text{T}}g(\bm{X}(t))\Big\}.
\end{align*}
As a result, $\frac{d}{dt}\mathcal{E}(\bm{\mathsf{X}}_{e})=-\text{tr}\Big\{\bm{\mathsf{X}}_{e}^{\text{T}}h(\bm{\mathsf{X}}_{e}^{\text{T}},\hat{\bm{\mathsf{X}}})\Big\}+o(\|\bm{\mathsf{X}}_{e}\|_{\text{F}}^{2})$. 

Using (\ref{eq:gradient-iteration-U}) \textendash{} (\ref{eq:gradient-iteration-V})
and the fact that $\bm{\mathsf{H}}=\hat{\bm{\mathsf{U}}}\hat{\bm{\mathsf{V}}}^{\text{T}}$,
it can be shown that 
\begin{align*}
\frac{1}{2}\text{tr}\Big\{\bm{\mathsf{X}}_{e}^{\text{T}}h(\bm{\mathsf{X}}_{e},\bm{X})\Big\} & =\text{tr}\bigg\{\bm{\mathsf{U}}_{e}^{\text{T}}\Big(\bm{W}\varodot\big(\bm{\mathsf{U}}_{e}\bm{V}^{\text{T}}\big)\Big)\bm{V}\\
 & \qquad+\bm{\mathsf{V}}_{e}^{\text{T}}\Big(\bm{W}^{\text{T}}\varodot\big(\bm{\mathsf{V}}_{e}\bm{U}^{\text{T}}\big)\Big)\bm{U}\bigg\}\\
 & \qquad\quad+\text{tr}\bigg\{\bm{\mathsf{U}}_{e}^{\text{T}}\Big(\bm{W}\varodot\big(\bm{U}\bm{\mathsf{V}}_{e}^{\text{T}}\big)\Big)\bm{V}\bigg\}\\
 & \qquad\qquad+\text{tr}\bigg\{\bm{\mathsf{V}}_{e}^{\text{T}}\Big(\bm{W}^{\text{T}}\varodot\big(\bm{V}\bm{\mathsf{U}}_{e}^{\text{T}}\big)\Big)\bm{U}\bigg\}.
\end{align*}

Note that under $\bm{\mathsf{H}}=\bm{H}$, we have $\bm{W}=\mathbf{1}_{N\times N}$.
In addition, if $\|\bm{\mathsf{V}}_{e}\|_{\text{F}}=o\big(\|\bm{\mathsf{U}}_{e}\|_{\text{F}}\big)$,
\emph{i.e.}, $\|\bm{\mathsf{V}}_{e}\|_{\text{F}}\ll\|\bm{\mathsf{U}}_{e}\|_{\text{F}}$,
then 
\begin{align*}
\frac{1}{2}\text{tr}\Big\{\bm{\mathsf{X}}_{e}^{\text{T}}h(\bm{\mathsf{X}}_{e},\bm{X})\Big\} & =\text{tr}\Big\{\bm{\mathsf{U}}_{e}^{\text{T}}\bm{\mathsf{U}}_{e}\bm{V}^{\text{T}}\bm{V}\Big\}+o\Big(\|\bm{\mathsf{U}}_{e}\|_{\text{F}}^{2}\Big)\\
 & =\text{tr}\Big\{\bm{\mathsf{U}}_{e}\bm{V}^{\text{T}}\bm{V}\bm{\mathsf{U}}_{e}\Big\}+o\Big(\|\bm{\mathsf{U}}_{e}\|_{\text{F}}^{2}\Big)\\
 & =\sum_{j=1}^{N}\bm{\mathsf{u}}_{j}^{(e)}\bm{V}^{\text{T}}\bm{V}\bm{\mathsf{u}}_{j}^{(e)\text{T}}+o\Big(\|\bm{\mathsf{U}}_{e}\|_{\text{F}}^{2}\Big)\\
 & \geq\sum_{j=1}^{N}\lambda_{K}\big(\bm{V}^{\text{T}}\bm{V}\big)\|\bm{\mathsf{u}}_{j}^{(e)}\|^{2}+o\Big(\|\bm{\mathsf{U}}_{e}\|_{\text{F}}^{2}\Big)
\end{align*}
where $\bm{\mathsf{u}}_{j}^{(e)}$ is the $j$th row vector of the
matrix $\bm{\mathsf{U}}_{e}$. As a result, we have 
\[
\frac{d}{dt}\mathcal{E}(\bm{\mathsf{X}}_{e})\leq-2\lambda_{K}(\hat{\bm{\mathsf{V}}}\hat{\bm{\mathsf{V}}}^{\text{T}})\|\bm{\mathsf{U}}_{e}\|_{\text{F}}^{2}+o\Big(\|\bm{\mathsf{U}}_{e}\|_{\text{F}}^{2}\Big)
\]
proving (\ref{eq:error-convergence-rate-U}).

For the case of $\|\bm{\mathsf{U}}_{e}\|_{\text{F}}=o\big(\|\bm{\mathsf{V}}_{e}\|_{\text{F}}\big)$,
the derivation to show (\ref{eq:error-convergence-rate-V}) is similar.
\hfill\IEEEQED

\bibliographystyle{IEEEtran}
%\bibliography{../Bibliography/IEEEabrv,../Bibliography/StringDefinitions,../My_reference,../Bibliography/ChenBibCV}

\end{document}